\documentclass[letterpaper]{scrartcl} 

\usepackage[]{brandl}

\newcommand{\nash}[1][]{\ifthenelse{\equal{#1}{}}{\ensuremath{\mathit{NE}}\xspace}{\ensuremath{\mathit{NE}(#1)}\xspace}}
\newcommand{\nashpure}[1][]{\ifthenelse{\equal{#1}{}}{\ensuremath{\mathit{NE}_{\mathrm{pure}}}\xspace}{\ensuremath{\mathit{NE}_{\mathrm{pure}}(#1)}\xspace}}
\newcommand{\nashfs}[1][]{\ifthenelse{\equal{#1}{}}{\ensuremath{\mathit{NE}_{\mathrm{fs}}}\xspace}{\ensuremath{\mathit{NE}_{\mathrm{fs}}(#1)}\xspace}}
\newcommand{\nashq}[1][]{\ifthenelse{\equal{#1}{}}{\ensuremath{\mathit{NE}_{\mathbb{Q}}}\xspace}{\ensuremath{\mathit{NE}_{\mathbb{Q}}(#1)}\xspace}}
\newcommand{\ce}[1][]{\ifthenelse{\equal{#1}{}}{\ensuremath{\mathit{CE}}\xspace}{\ensuremath{\mathit{CE}(#1)}\xspace}}
\newcommand{\ceconn}[1][]{\ifthenelse{\equal{#1}{}}{\ensuremath{\mathit{CE}_{\mathrm{conn}}}\xspace}{\ensuremath{\mathit{CE}_{\mathrm{conn}}(#1)}\xspace}}
\newcommand{\cce}[1][]{\ifthenelse{\equal{#1}{}}{\ensuremath{\mathit{CCE}}\xspace}{\ensuremath{\mathit{CCE}(#1)}\xspace}}
\newcommand{\eps}{\varepsilon}

\newcommand{\uni}{\mathrm{uni}}
\newcommand{\setm}{\setminus}
\newcommand{\timesdots}{\times\dots\times}

\title{Axioms for Correlated Equilibrium}

\author{Florian Brandl\thanks{\texttt{florian.brandl@uni-bonn.de}}\\ \large Institute for Microeconomics, University of Bonn}

\date{\today}

\begin{document}
	
\maketitle
\begin{abstract}
We characterize correlated equilibrium in finite normal-form games. Interpreting correlated strategies as action recommendations, we show that correlated equilibrium is the unique solution concept that never recommends a pure-strategy dominated action, treats payoff-equivalent actions interchangeably, and respects the sure-thing principle under uncertainty about payoffs and the correlation device. A parallel characterization identifies coarse correlated equilibrium among solution concepts that recommend dominant actions whenever they exist and treat payoff-equivalent actions as strongly interchangeable.
\end{abstract}

\section{Introduction}\label{sec:introduction}

Public randomization in strategic games induces correlation among the players' strategies.
Such correlated strategies enable outcomes that are infeasible when players randomize independently.
Game-theoretic analysis can thus draw from a larger pool of solution concepts when describing, predicting, or explaining outcomes of games. 
To guide the choice of the solution concept from among this wide range of possibilities, we take an axiomatic approach: we formulate axioms that require correlated strategies to be selected coherently across games and determine which solution concepts satisfy those.
We thereby obtain characterizations of correlated equilibrium and coarse correlated equilibrium.
The characterizations clarify which assumptions underlie each solution concept and can guide which one to use.

In our model, the set of players is fixed, and we consider normal-form games with a finite but variable set of actions for each player.
A correlated strategy is a distribution over action profiles, interpreted as a correlation device that recommends an action to each player.
A solution concept assigns to each game a set of correlated strategies.
Nash equilibrium, correlated equilibrium, and coarse correlated equilibrium are all examples of solution concepts.
The axioms we impose require that correlated strategies are selected coherently across games and that they respect rationality.

\emph{Consistency} addresses payoff uncertainty. Nature draws payoffs, and each draw yields a deterministic game.
If a correlated strategy is selected for every realization, consistency requires it to be selected for the lottery over realizations.
Rather than model lotteries explicitly, we identify them with convex combinations of the realized games, replacing random payoffs by their expectations.
In this sense, consistency applies the sure-thing principle of \citet{Sava54a} to strategic games.

\emph{Consequentialism} specifies how the solution concept responds when an action is duplicated.
Two actions of a player are \emph{clones} if every player's utility is the same under either action, regardless of the other players' actions.
In a two-player game, clones correspond to rows or columns that are identical in both players' payoff matrices.
When we add a clone, consequentialism allows the probability of any action profile involving the original action to be split in a fixed ratio between that profile and the profile in which the original action is replaced by its clone; all other action-profile probabilities stay fixed.
Requiring the same ratio across all affected profiles ensures that recommendations of either clone convey the same information about the other players' actions.\footnote{For strategy profiles with independent randomization, consequentialism is weaker than the invariance axiom of \citet{KoMe86a}, which also allows introducing actions that are convex combinations of existing ones.}

\emph{Rationality} rules out recommendations of strictly dominated actions: an action profile receives probability $0$ whenever some player's action is strictly dominated in pure strategies.
We also require the solution concept to be continuous and to return a nonempty, convex set of correlated strategies in each game.
Convex-valuedness keeps mixtures of plausible correlation devices plausible, so the solution concept does not exclude outcomes that arise from uncertainty about which device is used.
Like consistency, this is another instance of the sure-thing principle.\footnote{\Cref{sec:discussion} discusses the behavioral implications of the axioms in detail.}

In correlated equilibrium, following the recommendation is optimal for a Bayesian expected utility maximizer.
Our main result (\Cref{thm:ce}) shows that correlated equilibrium is the only solution concept satisfying these axioms.
Equivalently, these axioms are exactly the coherence requirements that pin down correlated equilibrium.
In particular, they encode Bayesian expected-utility maximization.

The second characterization strengthens consequentialism and weakens rationality.
Strong consequentialism drops the fixed-ratio requirement: when an action is cloned, probability can be split arbitrarily between the original and the clone.
Clones can therefore carry different information about the other players' strategies.
Correlated equilibrium violates strong consequentialism because changing the information carried by a recommendation can make it suboptimal to follow either clone when recommended.
Combined with our other axioms, this yields an impossibility.
Coarse correlated equilibrium, introduced by \citet{MoVi78a}, requires that each player's expected utility from always following the recommendation is at least as high as the utility from committing to a fixed action before observing the recommendation.
Because this incentive constraint does not depend on how probability is split between clones, coarse correlated equilibrium respects strong consequentialism.
It violates rationality, however, because it rules out only unconditional deviations and need not make following the recommendation optimal after it is observed.
On the other hand, it is optimal to commit to dominant actions whenever they exist, so coarse correlated equilibrium satisfies the weaker form of rationality that requires dominant actions to be recommended.
Our second result (\Cref{thm:cce}) characterizes coarse correlated equilibrium as the unique continuous and convex-valued solution concept that satisfies consistency, strong consequentialism, and weak rationality, and that selects at least one correlated strategy for each game.

Both results contribute to the recent literature on axiomatic characterizations of solution concepts. 
In addition to showing which solution concepts are forced upon us when committing to a set of axioms, their necessity parts clarify which properties one commits to when using the characterized solution concept.

\section{Related Work}\label{sec:related-work}

Most axiomatic characterizations for normal-form games treat only uncorrelated strategies, i.e., solution concepts that return profiles of independent strategies.
The closest work to ours is \citet{BrBr23a}, who show that Nash equilibrium is the only total solution concept satisfying consistency, consequentialism, and weak rationality.
Once correlation is allowed, these axioms no longer pin down a unique concept: Nash equilibrium, correlated equilibrium, and coarse correlated equilibrium all satisfy them.
Their proofs rely on independent randomization and do not extend to correlated strategies.
\citet{SSTW25a} introduce a narrow bracketing axiom: if two strategy profiles are selected in two independent games, then their product profile is selected when the games are played simultaneously.\footnote{As noted by \citet{SSTW25a}, consistency and consequentialism imply narrow bracketing for uncorrelated strategies.}
They show that any total solution concept satisfying narrow bracketing, anonymity, rationality, and monotonicity in expected payoffs is a refinement of Nash equilibrium.

Building on \citet{PeTi96a}, \citet{NPRV96a} characterize Nash equilibrium using utility maximization in one-player games and a consistency condition that varies the set of players.
Their consistency requires that any strategy profile returned for an $n$-player game is also returned for the $(n-k)$-player game obtained by fixing the strategies of $k$ players.
Because our model keeps the player set fixed and allows correlated strategies, their characterization does not apply here.
Other axiomatic work on Nash equilibrium includes characterizations of pure Nash equilibrium \citep{Voor19a}, Nash equilibrium for games with quasiconcave utility functions \citep{Salo92a}, and a choice-theoretic characterization \citep{Cres26a}.

\citet{BrBr17c} consider solution concepts for two-player zero-sum games that return a set of strategies for a single player instead of strategy profiles or correlated strategies.
They show that returning all maximin strategies is the coarsest such solution concept that satisfies consistency, consequentialism, and rationality.
While we consider two-player zero-sum games in our proofs, there is no methodological connection with their work.

Epistemic game theory studies what players need to know about other players to justify Nash equilibrium.
The players' knowledge is modeled using Bayesian belief hierarchies, which consist of a game and a set of types for each player; a type includes the action played by this type and a belief about the types of the other players \citep{Hars67a}. 
That is, players have probabilistic beliefs about other players' types, but play a deterministic action conditional on their type.
A rational player maximizes their expected payoff given their type.
\citet{AuBr95a} show that for two-player games, the beliefs of every pair of types constitute a Nash equilibrium if their beliefs and rationality are mutually known.
This result extends to games with more than two players if the beliefs are commonly known and admit a common prior. 
\citet{Bare09a}, \citet{Hell13a}, and \citet{BaTs14a} show that the results of \citet{AuBr95a} still hold under weaker common knowledge assumptions.
\citet{AuDr08a} study the payoff a player can expect in a game with common knowledge of rationality and a common prior. 
They characterize rational expectations via correlated equilibria of an extended game: it is rational to expect a given payoff precisely when it is a recommendation-contingent expected payoff in some correlated equilibrium of the doubled game, where each of the player's actions is replaced by two clones.

The literature on equilibrium refinements takes equilibrium play as given and singles out equilibria with additional properties.
It focuses on equilibria that are strategically stable and robust with respect to the representation of a game.
The closest connection to our work is that the invariance axiom of \citet{KoMe86a}---the equilibrium selection only depends on the reduced normal form of a game---implies consequentialism.
The invariance to embedding condition of \citet{GoWi12a} is a further strengthening that also considers variable sets of players.

\section{Preliminaries}\label{sec:prelims}

Let $U$ be an infinite universal set of actions, and let $\mathcal F(U)$ be the collection of nonempty finite subsets of $U$.
Fix the player set $N = \{1,\dots,n\}$.
For action sets $A_1,\dots,A_n\in\mathcal F(U)$, write $A = A_1\timesdots A_n$ for the set of action profiles.
An $n$-player normal-form game on $A$ is a function $G\colon A\to\mathbb R^n$; player $i$'s payoff at $a\in A$ is $G_i(a)$.
A correlated strategy is a probability distribution on $A$, and the set of correlated strategies is denoted by $\Delta(A)$.
Interpret a correlated strategy $p$ as a signal structure: a joint signal $a$ is drawn from $p$, and player $i$ observes $a_i$ as a recommendation to play $a_i$.
Throughout, \emph{action} means a pure strategy, \emph{strategy} a mixed strategy, and \emph{profile} a player-indexed vector.
 
We say that $G\colon A\to \mathbb R^n$ is a blow-up of $G'\colon A'\to\mathbb R^n$ if $G$ is obtained from $G'$ by replacing each action with one or more payoff-equivalent actions.
That is, there is a surjection $\phi = (\phi_1,\dots,\phi_n)\colon A\to A'$ with $\phi_i\colon A_i\rightarrow A_i'$ such that $G = G'\circ\phi$.
Thus, actions in $\phi_i^{-1}(a_i')$ are payoff-equivalent copies, called \emph{clones}, of $a_i'\in A_i'$, and $G$ is obtained from $G'$ by replacing $a_i'$ by $|\phi_i^{-1}(a_i')|$ clones.
A correlated strategy $p\in\Delta(A)$ induces a correlated strategy $\phi_*(p)= p\circ\phi^{-1}\in\Delta(A')$ with $\phi_*(p)(a') = \sum_{a\in \phi^{-1}(a')} p(a)$ for $a'\in A'$.\footnote{%
Blow-ups of games also appear in work on equilibrium refinements.
The invariance axiom of \citet{KoMe86a} uses a more permissive notion of blowing up.
There, $G'$ is a reduced form of $G$ if $G'$ is obtained from $G$ by deleting actions that are convex combinations of other actions.
\citet{ChGo06a} connect two refinements of correlated equilibrium---perfect correlated equilibrium and perfect direct correlated equilibrium (introduced by \citet{DhMe96a}) through blow-ups: every perfect correlated equilibrium is induced by a perfect direct correlated equilibrium in a blown-up game.
}

A solution concept $f$ maps every game $G$ to a set of correlated strategies $f(G)$ on the action profiles of $G$.
Call $f$ \emph{total} if $f(G)\neq \emptyset$ for each $G$, \emph{continuous} if it is an upper hemi-continuous correspondence, and \emph{convex-valued} if $f(G)$ is convex for each $G$.\footnote{
$f$ is upper hemi-continuous if for each $A\in \mathcal F(U)^n$, each sequence of games $(G^\ell)$ on $A$ converging to $G$, each sequence of correlated strategies $(p^\ell)$ with $p^\ell \in f(G^\ell)$ converging to $p$, we have $p\in f(G)$. 
}

Three common solution concepts are Nash equilibrium, correlated equilibrium, and coarse correlated equilibrium.
A \emph{Nash equilibrium} is a correlated strategy $p\in\Delta(A)$ that is a product distribution of $p_i\in\Delta(A_i)$, $i\in N$, such that $(p_1,\dots,p_n)$ is a Nash equilibrium of $G$ in the usual sense.\footnote{Since we work with correlated strategies throughout, it is more convenient to view a Nash equilibrium as a correlated strategy that is the product of independent strategies rather than as a strategy profile consisting of a strategy for each player.}

A \emph{correlated equilibrium} is a correlated strategy $p\in \Delta(A)$ such that for each $i \in N$ and $a_i,b_i\in A_i$,
\begin{align*}
	\sum_{a_{-i} \in A_{-i}} p(a_i,a_{-i}) \left( G_i(a_i,a_{-i}) - G_i(b_i,a_{-i})\right) \ge 0
\end{align*}
Thus, in a correlated equilibrium, it is optimal for each player to follow their recommendation when conditioning the distribution of other players' action profiles on the recommendation.
Nash equilibria are correlated equilibria where the players' recommendations are independent, so that conditioning on a player's recommendation does not change the distribution of other players' action profiles.

A \emph{coarse correlated equilibrium} is a correlated strategy $p\in\Delta(A)$ such that for each $i\in N$ and $b_i \in A_i$,
\begin{align*}
	\sum_{a \in A} p(a_i,a_{-i}) \left(G_i(a_i,a_{-i}) - G_i(b_i,a_{-i})\right) \ge 0
\end{align*}
Each player is weakly better off always following the recommendation than committing ex ante to a fixed action.
Each correlated equilibrium is a coarse correlated equilibrium since the former demands that deviations are not even profitable after observing the recommendation.

We denote by $\nash$, $\ce$, and $\cce$ the solution concepts that return the set of all Nash equilibria, correlated equilibria, and coarse correlated equilibria, respectively, for each game.

\section{The axioms}\label{sec:axioms}

Consistency requires that any correlated strategy returned in two games with the same action sets is also returned in any convex combination of the two games.
\begin{definition}[Consistency]\label{def:consistency}
	A solution concept $f$ satisfies consistency if for any two games $G,G'$ on $A$ and any $\lambda\in[0,1]$,
	\begin{align*}
		f(G)\cap f(G')\subseteq f(\lambda G + (1-\lambda) G').
	\end{align*}
\end{definition}
Consistency constrains how a solution concept responds to payoff uncertainty.
If a correlated strategy $p$ is returned for both $G$ and $G'$, then $p$ is also returned for a lottery that chooses between $G$ and $G'$.
We represent such lotteries by convex combinations, replacing random payoffs by their expectations.
In this sense, consistency applies the sure-thing principle of \citet{Sava54a} to strategic games.
\Cref{sec:discussion} analyzes the behavioral content of this axiom and the others.

In the game theory literature, consistency appears in characterizations of Nash equilibrium \citep{BrBr23a,KaKa24a} and of maximin strategies in two-player zero-sum games \citep{BrBr17c}.\footnote{\citet{KaKa24a} call the axiom the sure thing principle instead of consistency.}
Variants of consistency have been considered for single-player decision problems with uncertainty about the state of nature \citep[see, e.g.,][]{Cher54a,Miln54a,GiSc03a}.
Analogs of consistency feature prominently in axiomatic characterizations in social choice theory, where it relates the choices for different sets of voters to each other \citep[see, e.g.,][]{Smit73a,Youn75a,YoLe78a,Myer95b,Bran13a,LaSk21a}.
\citet{HaSe72a} use a variant of consistency (their Axiom 8) to characterize an extension of Nash's bargaining solution to bargaining under uncertainty.
The characterization of the Shapley value by \citet{Shap53c} also 
involves an additivity axiom (which he calls law of aggregation) that is similar in spirit to consistency.

Consequentialism imposes invariance to cloning actions.
When an action is replaced by two payoff-equivalent clones, probability mass can move only between action profiles that differ by which clone is used; all other action-profile probabilities remain unchanged.
This invariance has two natural interpretations, which lead to different definitions of consequentialism.
Consider the following two games and corresponding correlated strategies.
\[
G' =
\begin{pmatrix}
1,0 & 0,0\\
0,0 & 1,0
\end{pmatrix}
\quad
p' =
\begin{bmatrix}
	\frac12 & \frac12\\
	0 & 0
\end{bmatrix}
\quad\quad
G =
\begin{pmatrix}
1,0 & 0,0\\
1,0 & 0,0\\
0,0 & 1,0
\end{pmatrix}
\quad
p =
\begin{bmatrix}
\frac13        & \frac13\\
\frac16        & \frac16 \\
0        & 0
\end{bmatrix}
\quad
\tilde p =
\begin{bmatrix}
\frac12        & \frac16\\
0        & \frac13 \\
0        & 0
\end{bmatrix}
\]
Here, $G$ is obtained from $G'$ by cloning the row player's first action.
Both $p$ and $\tilde p$ induce $p'$ because each assigns probability $\frac12$ to the two clones of the action profiles in the first row of $G'$.
However, $p$ and $\tilde p$ distribute probability differently across the two clones.
Under $p$, a recommendation of either clone induces the same distribution over columns, and therefore the same belief about the column player's recommendation; under $\tilde p$, the induced beliefs differ.
Under the weaker notion, $p$ is returned for $G$ exactly when $p'$ is returned for $G'$.
Under the stronger notion, the same equivalence must also hold for $\tilde p$.

Suppose $G$ is a blow-up of $G'$ with surjection $\phi\colon A\to A'$. A correlated strategy $p\in\Delta(A)$ is \emph{clone-symmetric} if for every $i\in N$ and $a_i\in A_i$ there exists $\alpha_i(a_i) \in [0,1]$ such that for all $a_{-i} \in A_{-i}$, $p(a_i,a_{-i}) = \alpha_i(a_i) \sum_{b_i\in A_i\colon \phi_i(b_i)=\phi_i(a_i)} p(b_i,a_{-i})$.
Equivalently, any two clones of the same action in $G'$ either induce the same conditional distribution on $A_{-i}$ or receive probability $0$.
In the example, $p$ is clone-symmetric, whereas $\tilde p$ is not.

\begin{definition}[(Strong) consequentialism]\label{def:consequentialism}
	A solution concept $f$ satisfies consequentialism if for all games $G,G'$ such that $G$ is a blow-up of $G'$ with surjection $\phi$ and for each clone-symmetric $p\in\Delta(A)$,
	\begin{align*}
		p \in f(G) \text{ if and only if } \phi_*(p) \in f(G').
	\end{align*}
	Moreover, $f$ satisfies strong consequentialism if for all games $G,G'$ such that $G$ is a blow-up of $G'$ with surjection $\phi$ and all $p\in\Delta(A)$,
	\begin{align*}
		p \in f(G) \text{ if and only if } \phi_*(p) \in f(G').
	\end{align*}
\end{definition}
Consequentialism requires that a clone-symmetric correlated strategy is returned in $G$ if and only if its induced correlated strategy is returned in $G'$.
Equivalently, when adding a clone $a_i'$ of $a_i$ of player $i$, (i) for any $\alpha\in[0,1]$, the probability of any action profile involving $a_i$ can be split among its two clones so that the first clone receives an $\alpha$ fraction of the probability, and (ii) the probability of any action profile not involving $a_i$ remains the same.
Strong consequentialism requires that a correlated strategy is returned in $G$ if and only if its induced correlated strategy is returned in $G'$.
Equivalently, (i') the probability of an action profile including $a_i$ may be split arbitrarily among its two clones, and (ii) holds.
Under strong consequentialism, the two clones may induce different distributions over the action profiles of the players $N\setminus \{i\}$.
Strong consequentialism implies consequentialism, and consequentialism implies equivariance: relabeling a player's actions relabels the returned correlated strategies in the same way (see the appendix).
When restricting to strategies with independent randomization, both notions of consequentialism agree since such strategies are automatically clone-symmetric.

(Strong) consequentialism appears in the characterizations of \citet{BrBr17c,BrBr23a}; \citet{KaKa24a} uses the term strategy anonymity.
The invariance axioms of \citet{KoMe86a} and \citet{GoWi12a} strengthen consequentialism by also allowing convex combinations of existing actions.
Decision theory uses closely related invariance conditions as well.
For example, consequentialism corresponds to Postulate~6 (cloning a player's actions) together with Postulate~9 (cloning Nature's states, i.e., opponents' actions) in \citet{Cher54a}. Postulate~9 also appears as column duplication \citep{Miln54a} and as deletion of repetitious states \citep{ArHu72a,Mask79a}.
In social choice theory, \citet{Tide87a} introduced independence of clones, a parallel condition for cloning social alternatives \citep[see also][]{ZaTi89a,Bran13a}.

Rationality prescribes that actions that are dominated in pure strategies are never played. 
Weak rationality is weaker: it requires only that a dominant action is played with positive probability.
Formally, for $a_i,a_i'\in A_i$, $a_i$ \emph{dominates} $a_i'$ if $G_i(a_i,a_{-i}) > G_i(a_i',a_{-i})$ for each $a_{-i}\in A_{-i}$; in that case, $a_i'$ is \emph{dominated} by $a_i$. 
Action $a_i$ is \emph{dominant} if it dominates every other action in $A_i$.

\begin{definition}[(Weak) rationality]\label{def:rationality}
	A solution concept $f$ satisfies rationality if for each game $G$, each $p\in f(G)$, each $i\in N$, and each $a_i\in A_i$,
	\begin{align*}
		\text{if $a_i$ is dominated, then } p(a_i,a_{-i}) = 0 \text{ for each }a_{-i}\in A_{-i}.
	\end{align*}	
	A solution concept $f$ satisfies weak rationality if for each game $G$, each $p\in f(G)$, each $i\in N$, and each $a_i\in A_i$,
	\begin{align*}
		\text{if $a_i$ is dominant, then } p(a_i,a_{-i}) > 0 \text{ for some } a_{-i} \in A_{-i}.
	\end{align*}
\end{definition}
Both axioms refer only to pure-strategy dominance, so they do not constrain how players evaluate lotteries.
They also impose no assumptions about beliefs about opponents' rationality.
Rationality is equivalent to the strong domination axiom of \citet{Miln54a} and to Property~5 of \citet{Mask79a}; it is weaker than Postulate~2 of \citet{Cher54a}.

\section{Results}
\label{sec:results}

Correlated equilibrium satisfies consistency, consequentialism, and rationality.
In a correlated equilibrium, each player maximizes expected utility given the recommendation, and these axioms capture that behavior.
Among total, continuous, and convex-valued solution concepts, only \ce satisfies all three axioms.

\begin{theorem}
	\label{thm:ce}
	$\ce$ is the only total, continuous, and convex-valued solution concept that satisfies consistency, consequentialism, and rationality.
\end{theorem}

\Cref{thm:ce} makes two claims. First, any solution concept satisfying the axioms returns every correlated equilibrium. Second, any coarsening of \ce violates at least one axiom.
For example, \nash violates convex-valuedness and \cce violates rationality.
The first claim forces the selection of every correlated strategy consistent with Bayesian expected utility maximization, which defines correlated equilibrium.
The second claim rules out coarsenings that admit behavior that violates Bayesian rationality.
Equivalently, the theorem characterizes Bayesian expected utility maximization.
If consequentialism is replaced by strong consequentialism, no solution concept can satisfy the resulting list of axioms, because \ce fails strong consequentialism.

Coarse correlated equilibrium corresponds to players who either commit to a fixed action ex ante or always follow the recommendation.
Equivalently, players maximize expected utility without updating beliefs about opponents' strategies after receiving a recommendation.
Under either interpretation, strong consequentialism is natural: if a player commits before observing the recommendation, or ignores the information it contains, then it does not matter how probability is split between clones.
While coarse correlated equilibrium violates rationality, it satisfies weak rationality.\footnote{Consider the two-player game $G$ and correlated strategy $p$ below.
\[
G =
\begin{pmatrix}
4,0 & 0,0\\
2,0 & 2,0\\
3,0 & 3,0
\end{pmatrix}
\quad
p =
\begin{bmatrix}
\frac12 & 0\\
\frac14 & \frac14\\
0 & 0
\end{bmatrix}
\]
The row player's second action is dominated by the third, and $p$ is a coarse correlated equilibrium with positive probability on the second row.}
Committing to a dominant action ex ante is always profitable; hence any coarse correlated equilibrium plays dominant actions with probability $1$.
The second result characterizes \cce using strong consequentialism and weak rationality.

\begin{theorem}
	\label{thm:cce}
	$\cce$ is the only total, continuous, and convex-valued solution concept that satisfies consistency, strong consequentialism, and weak rationality.
\end{theorem}

Weak rationality rules out coarsenings of \cce, while refinements of \cce (such as \nash and \ce) tend to violate strong consequentialism.
Thus, \Cref{thm:cce} characterizes expected utility maximization when players can choose between committing to a fixed action ex ante and always following the recommendation.

\section{Proof outlines}\label{sec:proofs}

Throughout this section, fix a solution concept $f$ that is total, continuous, and convex-valued and satisfies consistency and either (i) strong consequentialism and weak rationality or (ii) consequentialism and rationality.
The proofs of \Cref{thm:ce} and \Cref{thm:cce} extensively use decomposition: write a game as a convex combination of simpler games with the same equilibrium; then use consistency to conclude that the equilibrium is selected in the original game.
Both theorems rely on two preliminary results.
\begin{enumerate}[label=(\roman*)]
	\item $f$ returns all pure Nash equilibria in all games (\Cref{prop:pure-nash}). 
	\item $f$ returns all Nash equilibria in all essentially two-player zero-sum games (\Cref{prop:zero-sum-two-player}).\footnote{\Cref{prop:zero-sum-two-player} proves this statement for games that are zero-sum after a positive affine transformation of one player's utility function, rather than honest zero-sum games. This difference is immaterial for the argument.}
\end{enumerate}
In an essentially two-player zero-sum game, two players play a zero-sum game; all other players are dummies, with constant utility functions and no effect on others' payoffs.

The proof of \Cref{prop:pure-nash} has three steps.
First, we show that only those action profiles can be played that survive iterated restriction to dominant actions: if a player has a dominant action, remove all their other actions, and recurse.
If only a single action profile survives, it is played with probability $1$.
Second, any game with a quasi-strict pure Nash equilibrium is a convex combination of games for which only this action profile survives.
The first part and consistency imply that quasi-strict pure Nash equilibria are played.
Third, we extend the claim to all pure Nash equilibria using continuity.

The proof of \Cref{prop:zero-sum-two-player} starts from two-player matching-pennies games.
They have enough symmetries to show that the unique mixed Nash equilibrium---the uniform distribution over all action profiles---is played. 
Any essentially two-player zero-sum game is a convex combination of blow-ups of matching pennies games.
Together, this gives that the uniform distribution over all action profiles is played in an essentially two-player zero-sum game whenever it is a Nash equilibrium.
The third step heavily uses consequentialism to extend this statement to Nash equilibria that are neither uniform nor of full support.

The third part shows that any solution concept that returns all Nash equilibria in essentially two-player zero-sum games has a stronger property: it returns all correlated equilibria in all games (\Cref{sec:nash-to-ce}).
At this point, we know that $f$ is a coarsening of \ce.

To prove \Cref{thm:ce}, assume consequentialism and rationality.
For any correlated strategy that is not a correlated equilibrium of a given game, we construct a second game with two properties: the same strategy is a correlated equilibrium, and some action that is played in equilibrium is dominated in a convex combination of both games.
As a coarsening of \ce, $f$ returns this strategy for the second game, and thus also for the convex combination by consistency.
This contradicts rationality, and shows that $f$ is a refinement of \ce.

To prove \Cref{thm:cce}, assume strong consequentialism and weak rationality.
The argument has two parts.
First, fix any coarse correlated equilibrium of any game. Strong consequentialism lets us reduce to a case in which the equilibrium has additional structure. We then express the game as a convex combination of games for which the same correlated strategy is a correlated equilibrium. Since $f$ is a coarsening of \ce, consistency implies that $f$ also returns the strategy in the original game, so $f$ is a coarsening of \cce.
Second, to show that $f$ is a refinement of \cce, we use the same type of construction as above, but now produce a dominant action that is never played, contradicting weak rationality.

\section{Discussion}\label{sec:discussion}

\begin{remark}[Independence of the axioms]\label{rem:independence}
	\Cref{thm:ce} and \Cref{thm:cce} cease to hold if any of totality, consistency, (strong) consequentialism, or (weak) rationality is omitted.
	Moreover, convex-valuedness is required for \Cref{thm:ce}.
	We show this with a series of examples.
	\begin{enumerate}[label=\textit{(\roman*)}]
	
		\item \emph{Totality:} 
		Return all correlated strategies that randomize only over pure Nash equilibria.
		This solution concept satisfies continuity, convex-valuedness, consistency, strong consequentialism, and rationality, but it is not total.
		
		\item \emph{Convex-valuedness:} $\nash$ satisfies totality, continuity, consistency, consequentialism, and rationality, but it violates convex-valuedness.
		Note that $\nash$ also violates strong consequentialism, and thus does not prove that convex-valuedness is required for \Cref{thm:cce}.

		\item \emph{Consistency:} Return all correlated strategies that randomize only over action profiles that survive iterated elimination of dominated strategies. This solution concept satisfies totality, continuity, convex-valuedness, strong consequentialism, and rationality, but it violates consistency.

		\item \emph{Consequentialism:} Return all correlated strategies that randomize only over action profiles for which each player's action is optimal against uniformly randomizing opponents. 
		This solution concept satisfies totality, continuity, convex-valuedness, consistency, and rationality, but it violates consequentialism.

		\item \emph{Weak rationality:} Return all correlated strategies that maximize the sum of the players' payoffs. 
		This solution concept satisfies totality, continuity, convex-valuedness, consistency, and strong consequentialism, but it violates weak rationality.
	
	\end{enumerate}
	It is open whether continuity is needed for either result, and whether convex-valuedness is needed for \Cref{thm:cce}.
	An intriguing open question is whether \nash and \ce are the only solution concepts that satisfy all axioms in \Cref{thm:ce} except for convex-valuedness.
\end{remark}

\begin{remark}[Equilibrium refinements]\label{rem:refinements}
	In line with Selten's trembling-hand perfection \citep{Selt75a}, various authors have proposed refinements of correlated equilibrium based on robustness to small trembles \citep{Myer86a,DhMe96a,LQS22a,HKP26a}.
	A starting point for future research is to use the axiomatic method to characterize refinements of correlated equilibrium. 
	In \Cref{thm:ce}, convex-valuedness and the consistency axiom force many correlated strategies into the solution sets, ruling out refinements.
	Thus, weakening or replacing convex-valuedness and consistency, and jointly strengthening other axioms (e.g., rationality to admissibility) is a promising direction. 
\end{remark}

\begin{remark}[Behavioral content of the axioms]\label{rem:behavioral-content}
	When interpreted as behavioral conditions, the axioms take a stance on how players react to recommendations.
	View the correlated strategies returned by a solution concept as those that a mediator can implement---each player follows their recommendation.
	Under this interpretation, the results characterize the implementable correlated strategies implied by the axioms.
	Since correlated equilibria are exactly the implementable outcomes when players are Bayesian expected utility maximizers, the axioms entail Bayesian expected utility maximization.
	
	\begin{enumerate}[label=(\roman*)]
		\item \emph{Consistency:} Whenever a given correlated strategy is implementable for any payoff realization of a random game, it is also implementable for the lottery over games.
		Since we treat lotteries as convex combinations of the realizations, consistency is closely tied to expected utility theory: it replaces payoff uncertainty by expected payoffs.
		The same logic underlies the sure-thing principle of \citet{Sava54a}: if a decision-maker would choose the same action in every state, then they choose it when the state is uncertain.
		Consistency imposes this sure-thing principle on players.
		
		\item \emph{(Strong) consequentialism:} 	
		From a mediator's perspective, the difference between the two axioms is this.Suppose the mediator can implement a correlated strategy $p'$ for a game $G'$, and $G$ is obtained from $G'$ by adding a clone $a_i'$ of $a_i$.
	To implement a correlated strategy for $G$, draw an action profile from $p'$; if $i$'s action in this profile is $a_i$, toss a coin to decide whether $i$'s recommendation is $a_i$ or $a_i'$; do not change any other recommendations.
	Consequentialism requires that this correlated strategy is implementable for $G$ if the mediator uses the same coin irrespective of the action profile of $i$'s opponents.
	Strong consequentialism demands implementability even if coins with different biases are used for different action profiles.
	
	Adding clones is akin to enlarging the mediator's message space. 
	Consequentialism requires that if two recommendations are payoff-equivalent and informationally equivalent, then a player follows one if and only if they follow the other. 
	In other words, a player's choice is independent of the label of an action.
	This is compatible with players who update their beliefs about others' recommendations conditional on their own recommendation.	
	By contrast, strong consequentialism demands that players follow payoff-equivalent recommendations even if those induce different beliefs.
	This forces players to ignore the information contained in their recommendations, and is in line with committing to an action ex ante.
	The two axioms thus differ in what they assume about a player's arsenal of behavioral strategies.

		\item \emph{(Weak) rationality:} Both axioms state that players prefer higher \emph{sure} payoffs to lower ones, 
		and thus avoid assumptions about risk attitudes.
		Weak rationality allows implementability even when the mediator sometimes recommends dominated actions.
		If players can only choose between following the mediator's recommendation and committing to a fixed action ex ante, following can be optimal despite occasional dominated recommendations.
		When a dominant action exists, however, the only implementable recommendation is to play it.
		Weak rationality is thus more natural than rationality when players have to commit ex ante.
		
		\item \emph{Convex-valuedness:} 
		It states that any convex combination of implementable correlated strategies is implementable. 
		Such a convex combination arises if the mediator covertly tosses a coin to decide from which correlated strategy the recommendations are drawn.
		From the players' perspective, this introduces uncertainty about the distribution of the other players' recommendations.
		Convex-valuedness is thus another instance of the sure-thing principle: if a player follows their recommendation for either of two beliefs about others' recommendations, they also do so for any mixture of the two beliefs.
		The difference from random games is that uncertainty enters through the information about other players rather than through payoffs.
	\end{enumerate}
	
\end{remark}

\section*{Acknowledgments}
The author thanks Felix Brandt for numerous discussions on the topic.
The author also acknowledges support by the DFG under the Excellence Strategy {EXC-2047}.

\pagebreak
\appendix
\section*{APPENDIX}\label{sec:appendix}

\section{Preliminaries}\label{sec:appendix-prelims}

For finite sets $A,B$ with $B\subseteq A$, $\chi_B\in\{0,1\}^A$ denotes the indicator of $B$; the standard unit vector at $a \in A$ is $\chi_a$.
For $x,y \in \mathbb R^A$, $\supp(x)=\{a\in A\colon x(a)\neq 0\}$ is the support of $x$, and $\langle x,y \rangle = \sum_{a\in A} x(a)y(a)$ is the inner product of $x$ and $y$.
We denote by $\Pi(A)$ the set of permutations of $A$.

For $S\subseteq N$, $A\in\mathcal F(U)^n$, and $a\in U^N$, we write $A_S = \prod_{i\in S} A_i$ and $a_S=(a_i)_{i\in S}$, as well as $A_{-S} = A_{N\setminus S}$ and $a_{-S}=a_{N\setminus S}$.
For $x\in\mathbb R^A$, $x(\cdot,a_{-S}) \in \mathbb R^{A_S}$ is the restriction of $x$ to the coordinates in $S$ when fixing the remaining coordinates to $a_{-S}$.
The marginal of $p\in\Delta(A)$ with respect to $S$ is $p_S = \sum_{a_{-S} \in A_{-S}} p(\cdot,a_{-S}) \in \Delta(A_S)$.
We write $\uni(B) \in \Delta(A)$ for the uniform distribution on $B\subseteq A$. 

If $G$ is a game on $A\in\mathcal F(U)^n$ and $\pi=(\pi_1,\dots,\pi_n)$ with $\pi_i\in\Pi(A_i)$, then $G\circ\pi$ and $p\circ\pi$ denote the relabelings obtained by applying $\pi$ coordinatewise.
A solution concept is equivariant if relabeling the actions of a game results in the same relabeling among the returned correlated strategies.
\begin{definition}
	[Equivariance]
	A solution concept satisfies equivariance if for each game $G$ on $A$ and each $\pi = (\pi_1,\dots,\pi_n)$ with $\pi_i\in\Pi(A_i)$, $f(G\circ\pi)=f(G)\circ\pi$.
\end{definition}
Consequentialism implies equivariance: $G\circ\pi$ is a blow-up of $G$ with surjection $\pi$. 
We frequently apply equivariance to correlated strategies where the probabilities of all supported action profiles are the same, and the permutation maps each supported action profile to another supported action profile.
This gives a new game for which the same correlated strategy is returned.
 
A solution concept is positively homogeneous if it is invariant under scaling the payoffs of all players by the same positive constant.
\begin{definition}
	[Positive homogeneity]
	A solution concept $f$ is positively homogeneous if for each game $G$ and each $\alpha > 0$, $f(G) = f(\alpha G)$.
\end{definition}
Any total, continuous, and convex-valued solution concept that satisfies consistency and consequentialism contains a positively homogeneous solution concept that inherits any of our axioms from $f$.
For any such solution concept $f$, define $\tilde f$ for each game $G$ by
\begin{align*}
	\tilde f(G) = \bigcap_{\alpha \ge 1} f(\alpha G)
\end{align*}
and note that $\tilde f$ is a refinement of $f$.

\begin{lemma}[Homogeneous refinement]\label{lem:homogeneous-core}
	Let $f$ be a total, continuous, and convex-valued solution concept that satisfies consistency and consequentialism.
	Then, $\tilde f$ is total, continuous, convex-valued, positively homogeneous, and satisfies consistency and consequentialism.
	Moreover, if $f$ satisfies strong consequentialism or (weak) rationality, then $\tilde f$ also satisfies that axiom.
\end{lemma}
\begin{proof}
	First, we prove that $\tilde f$ is total.
	Let $G$ be a game on $A$, and let $G^0$ be the game on $A$ such that $G_i\equiv 0$ for each $i\in N$.
	Note that all actions of all players are clones in $G^0$.
	Hence, $\chi_a \in f(G^0)$ for each $a \in A$ by totality and consequentialism.
	Convex-valuedness then implies that $f(G^0) = \Delta(A)$.
	Then for each $\alpha \ge 1$, $$f(\alpha G) = f(\alpha G) \cap f(G^0) \subseteq f(\frac1\alpha \alpha G + (1-\frac1\alpha)G^0) = f(G)$$ by consistency, and so $f(\alpha G)$ is a non-increasing family of sets.
	Totality and continuity imply that $f(\alpha G)$ is nonempty and closed for each $\alpha$.
	Hence, $\tilde f(G)$ is nonempty.
	
	Second, we prove that $\tilde f$ is continuous.
	Let $(G^k)_{k\in\mathbb N}$ be a sequence of games on $A$ converging to $G$.
	For each $k\in\mathbb N$, let $p^k \in \tilde f(G^k)$, and assume that $(p^k)_{k\in\mathbb N}$ converges to $p\in\Delta(A)$.
	We need to show that $p\in \tilde f(G)$.
	To this end, it suffices to show that $p\in f(\alpha G)$ for each $\alpha \ge 1$.
	Now, for each $\alpha \ge 1$ and each $k$, $p^k \in f(\alpha G^k)$ by definition of $\tilde f$.
	Hence, since $f$ is continuous and $\alpha G^k$ converges to $\alpha G$ for each $\alpha\ge 1$, $p\in f(\alpha G)$.
	Thus, $p\in\tilde f(G)$.
	
	The proof of the remaining claims of the lemma is straightforward.
\end{proof}

The second lemma states that a total, positively homogeneous, and convex-valued solution concept that satisfies consistency and consequentialism is invariant under adding a constant to the utility function of a player.
\begin{lemma}
	[Adding constants]
	\label{lem:adding-constants}
	Let $f$ be a total, positively homogeneous, and convex-valued solution concept that satisfies consistency and consequentialism.
	Let $\tilde G$ be a game on $A$ such that $\tilde G_i$ is constant for each $i\in N$.
	Then, for each game $G$ on $A$, $f(G) = f(G + \tilde G)$.
\end{lemma}
\begin{proof}
	By the argument in the proof of \Cref{lem:homogeneous-core} (where $f(G^0)=\Delta(A)$), we have $f(\tilde G)=\Delta(A)$.	
	Consistency and positive homogeneity then imply that $f(G) = f(G) \cap f(\tilde G)\subseteq f(\frac12 G + \frac12 \tilde G) = f(G + \tilde G)$, and $f(G + \tilde G) = f(G + \tilde G) \cap f(-\tilde G) \subseteq f(\frac12 (G + \tilde G) + \frac12 (-\tilde G)) = f(G)$.	
\end{proof}

\section{Pure Nash equilibria}

We show that every total and continuous solution concept satisfying consistency and either (i) strong consequentialism and weak rationality or (ii) consequentialism and rationality returns all pure Nash equilibria.
Let $\nashpure$ denote the solution concept that returns all pure strategy Nash equilibria, let $\nashfs$ return all full support Nash equilibria, and let $\nashq$ return all Nash equilibria that assign rational-valued probabilities to all action profiles. For each game $G$ on $A$,
\begin{align*}
	\nashpure(G) &= \nash(G)\cap \{\chi_a\colon a \in A\}\\
	\nashfs(G) &= \nash(G) \cap \{p\in\Delta(A)\colon \supp(p) = A\}\\
	\nashq(G) &= \nash(G) \cap \mathbb Q^A.
\end{align*}
A Nash equilibrium $p = p_1\otimes\dots\otimes p_n\in\nash(G)$ is quasi-strict if for each $i\in N$, each $a_i \in \supp(p_i)$, and each $b_i\in A_i\setminus \supp(p_i)$, $G_i(a_i,p_{-i}) > G_i(b_i,p_{-i})$.

The first lemma shows that an action that is dominant in a subgame containing the support of a returned correlated strategy is played with probability $1$ in that correlated strategy.

\begin{lemma}[Restriction to dominant actions]\label{lem:support}
	Let $f$ be a solution concept satisfying consistency and either
	\begin{enumerate}[label=(\roman*),ref=(\roman*)]
		\item\label{case:support-strong} strong consequentialism and weak rationality, or
		\item\label{case:support-rational} consequentialism and rationality.
	\end{enumerate}
	Let $G$ be a game on $A$, $p\in f(G)$, and $\tilde A = \tilde A_1\timesdots \tilde A_n\subseteq A$ such that $\supp(p)\subseteq \tilde A$.
	Assume there are $j\in N$ and $b_j\in A_j$ such that
	\[
		G_j(b_j,a_{-j}) > G_j(a_j,a_{-j})
	\]
	for each $a_j\in A_j\setm\{b_j\}$ and each $a_{-j}\in \tilde A_{-j}$.
	Then, $\supp(p) \subseteq \{b_j\}\times \tilde A_{-j}$.
\end{lemma}
\begin{proof}
	Assume for contradiction that $\supp(p) \not\subseteq \{b_j\}\times \tilde A_{-j}$.
	Without loss of generality, $j = 1$ and $N\setm\{j\} = \{2,\dots,n\}$.
	Fix $a^* \in \supp(p)$ with $a^*_1\neq b_1$.
	
	We first show that, in case~\ref{case:support-strong}, we may assume without loss of generality that
	\[
		p(b_1,\cdot)\equiv 0.
	\]
	If $p(b_1,\cdot)\equiv 0$, there is nothing to prove.
	Thus, assume that case~\ref{case:support-strong} holds and that $p(b_1,\cdot)\not\equiv 0$.
	Let
	\[
		S=\{a_{-1}\in A_{-1}\colon p(b_1,a_{-1})>0\}.
	\]
	For each $a_{-1}\in S$, choose $m_{a_{-1}}\in\mathbb N$ so large that, with
	\[
		\delta_{a_{-1}}=\frac{p(b_1,a_{-1})}{m_{a_{-1}}},
	\]
	we have
	\[
		\sum_{a_{-1}\in S}\delta_{a_{-1}} < p(a^*).
	\]
	This is possible because $S$ is finite and $p(a^*)>0$.
	
	Let $B\in\mathcal F(U)$ be disjoint from $\bigcup_{i\in N} A_i$ and contain the following labels.
	For each $a_{-1}\in S$, let
	\[
		h^1_{a_{-1}},\dots,h^{m_{a_{-1}}}_{a_{-1}}
	\]
	be labels mapped to $(b_1,a_{-1})$, each with weight $\delta_{a_{-1}}$.
	For each $a_{-1}\in S$, let $\ell_{a_{-1}}$ be a label mapped to $a^*$, also with weight $\delta_{a_{-1}}$.
	Add one further label mapped to $a^*$ with weight
	\[
		p(a^*)-\sum_{a_{-1}\in S}\delta_{a_{-1}},
	\]
	and, for every remaining
	\[
		a\in\supp(p)\setm\bigl(\{a^*\}\cup(\{b_1\}\times S)\bigr),
	\]
	add one label mapped to $a$ with weight $p(a)$.
	Thus, there are a map $\sigma\colon B\to\supp(p)$ and strictly positive weights $(w_b)_{b\in B}$ such that
	\[
		\sum_{b\in B:\,\sigma(b)=a} w_b = p(a)
	\]
	for every $a\in A$.
	
	For each $i\in N$, let $\hat A_i=A_i\cup B$, and define $\phi_i\colon \hat A_i\to A_i$ by
	\[
		\phi_i(a_i)=a_i \quad\text{for each }a_i\in A_i,
		\qquad
		\phi_i(b)=\sigma(b)_i \quad\text{for each }b\in B.
	\]
	Let $\hat A=\hat A_1\timesdots\hat A_n$, and let $\hat G=G\circ\phi$.
	Then $\hat G$ is a blow-up of $G$.
	Define $\hat p\in\Delta(\hat A)$ by
	\[
		\hat p(b,\dots,b)=w_b \quad\text{for each } b\in B,
	\]
	and $\hat p(\hat a)=0$ for all other $\hat a\in\hat A$.
	Then $\phi_*(\hat p)=p$, so strong consequentialism implies $\hat p\in f(\hat G)$.
	
	Let
	\[
		H=\{h^r_{a_{-1}}\colon a_{-1}\in S,\ r\in\{1,\dots,m_{a_{-1}}\}\}.
	\]
	For $h=h^r_{a_{-1}}\in H$, write $\ell(h)=\ell_{a_{-1}}$.
	Let $\pi^h$ be the permutation of $\hat A$ that, for every player, transposes the two labels $h$ and $\ell(h)$ and fixes every other action.
	Since $w_h=w_{\ell(h)}$, we have $\hat p\circ \pi^h=\hat p$.
	Strong consequentialism implies consequentialism, and consequentialism implies equivariance; hence $\hat p\in f(\hat G\circ\pi^h)$ for each $h\in H$.
	Let
	\[
		\bar G=\frac1{|H|+1}\left(\hat G+\sum_{h\in H}\hat G\circ\pi^h\right).
	\]
	Repeated consistency implies $\hat p\in f(\bar G)$.
	
	Let $\hat{\tilde A}_i=B$ for each $i\in N$.
	We claim that, in $\bar G$, action $b_1$ strictly dominates every action in $\hat A_1\setm\{b_1\}$ against $\hat{\tilde A}_{-1}=B^{n-1}$.
	Fix $c_1\in\hat A_1\setm\{b_1\}$ and $c_{-1}\in B^{n-1}$.
	For each
	\[
		\pi\in\{\mathrm{id}\}\cup\{\pi^h\colon h\in H\},
	\]
	let
	\[
		z_{-1}=\phi_{-1}(\pi_{-1}(c_{-1})).
	\]
	Then $z_{-1}\in\tilde A_{-1}$, because every label in $B$ is mapped by $\sigma$ to an action profile in $\supp(p)\subseteq\tilde A$, and each $\pi^h$ only exchanges labels in $B$.
	Moreover,
	\[
		(\hat G\circ\pi)_1(b_1,c_{-1})
		=
		G_1(b_1,z_{-1}),
	\]
	while
	\[
		(\hat G\circ\pi)_1(c_1,c_{-1})
		=
		G_1(\phi_1(\pi_1(c_1)),z_{-1}).
	\]
	If $\phi_1(\pi_1(c_1))\neq b_1$, the former term is strictly larger than the latter by the dominance assumption in $G$.
	If $\phi_1(\pi_1(c_1))=b_1$, the two terms are equal.
	Thus every summand in the definition of $\bar G$ gives a weak advantage to $b_1$.
	This advantage is strict in the identity summand whenever $\phi_1(c_1)\neq b_1$.
	The only remaining possibility is that $\phi_1(c_1)=b_1$ and $c_1\neq b_1$.
	By construction, this implies $c_1\in H$.
	In the summand corresponding to $\pi^{c_1}$, we have
	\[
		\phi_1(\pi^{c_1}_1(c_1))=a^*_1\neq b_1,
	\]
	so the advantage is strict in that summand.
	Hence $b_1$ strictly dominates every action in $\hat A_1\setm\{b_1\}$ against $\hat{\tilde A}_{-1}$.
	
	Finally, $\supp(\hat p)\subseteq\hat{\tilde A}$, $\hat p(b_1,\cdot)\equiv0$, and
	\[
		\supp(\hat p)\not\subseteq \{b_1\}\times\hat{\tilde A}_{-1}.
	\]
	Thus, in case~\ref{case:support-strong}, replacing $(G,A,\tilde A,p)$ by $(\bar G,\hat A,\hat{\tilde A},\hat p)$ preserves the hypotheses and the contradiction assumption, and additionally gives $p(b_1,\cdot)\equiv0$.
	Suppressing hats, we assume from now on in case~\ref{case:support-strong} that $p(b_1,\cdot)\equiv0$.
	
	We now proceed in both cases.
	Since strong consequentialism implies consequentialism, $f$ satisfies consequentialism in case~\ref{case:support-strong} as well as in case~\ref{case:support-rational}.
	We successively construct games
	\[
		G = G^1, G^2,\dots, G^n
	\]
	on $A$ so that $b_1$ is dominant when $G^i$ is restricted to $A_{[i]}\times \tilde A_{-[i]}$, and $p\in f(G^i)$, for each $i\in[n]$.
	
	For $i=1$, this holds by assumption.
	Now let $i>1$ and assume the statement holds for all smaller values of $i$.
	If $A_i=\tilde A_i$, set $G^i=G^{i-1}$.
	Then $b_1$ is dominant when $G^i$ is restricted to $A_{[i]}\times\tilde A_{-[i]}$, and $p\in f(G^i)$.
	
	It remains to consider the case $A_i\setm\tilde A_i\neq\emptyset$.
	By the induction hypothesis, there is $\eps>0$ such that
	\[
		G_1^{i-1}(b_1,a_{-1})
		\ge
		G_1^{i-1}(a_1,a_{-1})+\eps
	\]
	for each $a_1\in A_1\setm\{b_1\}$ and each $a_{-1}\in A_{-1}$ with $a_{-[i-1]}\in\tilde A_{-[i-1]}$.
	Let
	\[
		\mu=\max_{a,b\in A} G_1^{i-1}(a)-G_1^{i-1}(b).
	\]
	Let $B_i\in\mathcal F(U)$ be disjoint from $A_i$ with
	\[
		|B_i|>\frac{\mu}{\eps}\,|A_i\setm\tilde A_i|,
	\]
	and let $\hat A_i=A_i\cup B_i$.
	Fix some $\hat a_i\in\tilde A_i$, and let $\hat G$ be the game on
	\[
		A_1\timesdots\times A_{i-1}\times\hat A_i\times A_{i+1}\timesdots\times A_n
	\]
	obtained from $G^{i-1}$ by introducing the actions in $B_i$ as clones of $\hat a_i$.
	Equivalently, $\hat G$ is a blow-up of $G^{i-1}$ with surjection $\hat\phi$ such that $\hat\phi_k$ is the identity for $k\neq i$, and $\hat\phi_i$ is the identity on $A_i$ and maps every action in $B_i$ to $\hat a_i$.
	Viewing $p$ as a correlated strategy on the enlarged action set with zero probability on profiles involving actions in $B_i$, consequentialism implies $p\in f(\hat G)$.
	
	Let $\hat\Pi_i\subseteq\Pi(\hat A_i)$ be the set of permutations of $\hat A_i$ that fix $\tilde A_i$ pointwise.
	For each $k\neq i$, let $\hat\Pi_k$ contain only the identity permutation of $A_k$.
	Let
	\[
		\hat\Pi=\hat\Pi_1\timesdots\hat\Pi_n.
	\]
	Define
	\[
		\bar G=\frac1{|\hat\Pi|}\sum_{\pi\in\hat\Pi}\hat G\circ\pi.
	\]
	All actions in $\hat A_i\setm\tilde A_i$ are clones of each other in $\bar G$.
	Moreover, $b_1$ is dominant when $\bar G$ is restricted to
	\[
		A_{[i-1]}\times \hat A_i\times \tilde A_{-[i]}.
	\]
	Indeed, if the $i$th action lies in $\tilde A_i$, the dominance margin is at least $\eps$ in every summand.
	If the $i$th action lies in $\hat A_i\setm\tilde A_i$, then, in the average over $\hat\Pi_i$, at least $|B_i|$ images are clones of $\hat a_i$ and therefore give margin at least $\eps$, while at most $|A_i\setm\tilde A_i|$ images can reduce the margin, each by at most $\mu$.
	The choice of $|B_i|$ makes the average margin strictly positive.
	
	For each $\pi\in\hat\Pi$, the permutation $\pi$ fixes every action profile in $\supp(p)$.
	Hence $p\circ\pi=p$.
	By equivariance, $p\in f(\hat G\circ\pi)$ for each $\pi\in\hat\Pi$.
	Repeated consistency therefore implies $p\in f(\bar G)$.
	
	Lastly, let $G^i$ be such that $\bar G$ is a blow-up of $G^i$ with surjection $\bar\phi$ such that $\bar\phi_k$ is the identity for each $k\neq i$, and $\bar\phi_i$ is the identity on $A_i$ and maps each action in $B_i$ to some fixed action in $A_i\setm\tilde A_i$.
	That is, $G^i$ is obtained from $\bar G$ by removing the actions in $B_i$.
	Since $p$ assigns probability zero to all profiles involving actions in $\hat A_i\setm\tilde A_i$, consequentialism implies $p\in f(G^i)$.
	Moreover, $b_1$ is dominant when $G^i$ is restricted to $A_{[i]}\times\tilde A_{-[i]}$.
	This proves the induction step.
	
	Thus, $p\in f(G^n)$ and $b_1$ is dominant in $G^n$.
	In case~\ref{case:support-strong}, this contradicts weak rationality, because $p(b_1,\cdot)\equiv0$.
	In case~\ref{case:support-rational}, this contradicts rationality, because $a^*\in\supp(p)$, $a^*_1\neq b_1$, and every action in $A_1\setm\{b_1\}$ is dominated by $b_1$ in $G^n$.
\end{proof}

Given any game, consider the subgame obtained by the following process: if some player has a dominant action, eliminate all of their other actions; repeat this step with the reduced game until no more eliminations are possible.
We say that an action profile survives iterated restriction to dominant actions if it is not eliminated during this process.
\Cref{lem:support} shows that a solution concept satisfying consistency and either (i) strong consequentialism and weak rationality or (ii) consequentialism and rationality returns only correlated strategies that are supported on action profiles that survive iterated restriction to dominant actions.

\begin{lemma}[Iterated restriction to dominant actions]\label{lem:iterated-dominant}
	Let $f$ be a solution concept satisfying consistency and either
	\begin{enumerate}[label=(\roman*),ref=(\roman*)]
		\item strong consequentialism and weak rationality, or
		\item consequentialism and rationality.
	\end{enumerate}
	Let $G$ be a game on $A$ and let $\tilde A = \tilde A_1\timesdots \tilde A_n\subseteq A$ be the set of action profiles that survive iterated restriction to dominant actions.
	Then, $f(G) \subseteq\Delta(\tilde A)$.
\end{lemma}
\begin{proof}
	This follows from repeated application of \Cref{lem:support}.
\end{proof}

If only a single action profile survives iterated restriction to dominant actions, this profile is the unique Nash equilibrium of the game.
\Cref{lem:iterated-dominant} implies that any total solution concept 	satisfying consistency and either (i) strong consequentialism and weak rationality or (ii) consequentialism and rationality returns this profile and nothing else.
We show that any game with a pure quasi-strict Nash equilibrium is a convex combination of games for which only this profile survives iterated restriction to dominant actions.
Together, this proves that all pure quasi-strict Nash equilibria are returned.

\begin{lemma}[Quasi-strict pure Nash equilibria I]\label{lem:decomposition-into-dominant-solvable}
	Let $G$ be a game on $A$ and $\tilde a\in A$ such that $\chi_{\tilde a}$ is a pure quasi-strict Nash equilibrium of $G$.
	Then, there are games $G^1,\dots,G^n$ on $A$ such that $\tilde a$ is the only action profile that survives iterated restriction to dominant actions in $G^i$ for each $i\in[n]$, and $G$ is a convex combination of $G^1,\dots,G^n$.
\end{lemma}
\begin{proof}
	For each $j\in N$, let $\tilde A^j = A_j\times\{\tilde a_{-j}\}$, and let $\hat G^j$ be a game on $A$ so that for each $i\in N\setm\{j\}$, $\tilde a_i$ dominates all actions in $A_i\setm \{\tilde a_i\}$, and for each $i\in N$ and $a\in \bigcup_{k\in N}\tilde A^k$,
	\begin{align*}
		\hat G^j_i(a) = G_i(a)
	\end{align*}
	This is feasible since $\tilde a$ is a pure quasi-strict Nash equilibrium of $G$.
	For the same reason, $\tilde a_j$ dominates all actions in $A_j\setm\{\tilde a_j\}$ if $\hat G^j$ is restricted to $\tilde A^j$.
	Note that, by construction, for each $i\in N$ and $a\in \bigcup_{k\in N} \tilde A^k$,
	\begin{align}
		\sum_{j\in N} \frac1n\cdot \hat G^j_i(a) = G_i(a).\label{eq:equaltog}
	\end{align}
	
	Fix $j\in N$, and let $G^j$ be the game on $A$ so that for each $i\neq j$, $G^j_i = \hat G^j_i$, and
	\begin{align*}
		G^j_j(a) =
		\begin{cases}
			\hat G^j_j(a) &\text{ for $a\in \bigcup_{k\in N} \tilde A^k$, and}\\
			n\cdot G_j(a) - \sum_{i\neq j} \hat G^i_j(a) &\text{ for $a\in A\setminus\bigcup_{k\in N} \tilde A^k$.}
		\end{cases}
	\end{align*}
	Observe that \eqref{eq:equaltog} also holds with $G^j$ in place of $\hat G^j$ since both games agree on $\bigcup_{k\in N}\tilde A^k$.
	Moreover, for each $i\neq j$, $\tilde a_i$ dominates all actions in $A_i\setm\{\tilde a_i\}$, and $\tilde a_j$ dominates all actions in $A_j\setm\{\tilde a_j\}$ if $G^j$ is restricted to $\tilde A^j$.
	Hence, $\tilde a$ is the only action profile in $G^j$ that survives iterated restriction to dominant actions.
	Lastly, for each $j\in N$ and each $a\in A\setminus \bigcup_{k\in N} \tilde A^k$,
	\begin{align*}
		\sum_{i\in N} G^i_j(a) &= G^j_j(a) + \sum_{i\neq j} G^i_j(a)\\
		&= n\cdot G_j(a) - \sum_{i\neq j} \hat G^i_j(a) + \sum_{i\neq j} G^i_j(a)\\
		&= n\cdot G_j(a) - \sum_{i\neq j} G^i_j(a) + \sum_{i\neq j} G^i_j(a)\\
		&= n\cdot G_j(a).
	\end{align*}
	For $a\in\bigcup_{k\in N} \tilde A^k$, the same equality follows from \eqref{eq:equaltog}.
	Hence, $G = \sum_{i\in N} \frac1n G^i$ as required.
\end{proof}

\begin{lemma}[Quasi-strict pure Nash equilibria II]\label{lem:quasi-strict-pure-nash}
	Let $f$ be a total solution concept satisfying consistency and either
	\begin{enumerate}[label=(\roman*),ref=(\roman*)]
		\item strong consequentialism and weak rationality, or
		\item consequentialism and rationality.
	\end{enumerate}
	Let $G$ be a game on $A$ and $\tilde a\in A$ such that $\chi_{\tilde a}$ is a pure quasi-strict Nash equilibrium of $G$.
	Then, $\chi_{\tilde a}\in f(G)$.
\end{lemma}
\begin{proof}
	Let $G^1,\dots,G^n$ be the games promised by \Cref{lem:decomposition-into-dominant-solvable}.
	For each $j\in N$, $\tilde a$ is the only action profile in $G^j$ that survives the iterated restriction to dominant actions.
	Hence, by \Cref{lem:iterated-dominant} and the assumption that $f$ is total, $\chi_{\tilde a}\in f(G^j)$ for each $j\in N$.
	Since $G = \sum_{j\in N} \frac1n\cdot G^j$, consistency implies that $\chi_{\tilde a}\in f(G)$. 
\end{proof}

For any pure Nash equilibrium in any game $G$, there is a sequence of games converging to $G$ in which that equilibrium is pure and \emph{quasi-strict}.
Hence, any continuous solution concept that returns pure quasi-strict Nash equilibria returns all pure Nash equilibria.
We thus have the following immediate consequence of \Cref{lem:quasi-strict-pure-nash}.

\begin{proposition}[Pure Nash equilibria]
	\label{prop:pure-nash}
	Let $f$ be a total and continuous solution concept satisfying consistency and either
	\begin{enumerate}[label=(\roman*),ref=(\roman*)]
		\item strong consequentialism and weak rationality, or
		\item consequentialism and rationality.
	\end{enumerate}
	Then, $\nashpure \subseteq f$.
\end{proposition}

\section{Two-player zero-sum games}\label{sec:two-player-zero-sum}

We show that every total, continuous, and convex-valued solution concept that satisfies consistency and either (i) strong consequentialism and weak rationality or (ii) consequentialism and rationality returns all Nash equilibria for each zero-sum game in which all but two players are dummy players.
Abusing terminology, we say that a game is zero-sum if there exist positive affine transformations of the players' utility functions such that for each action profile, the sum of all players' payoffs is $0$.\footnote{Any game that is zero-sum in this sense is strategically zero-sum as defined by \citet{MoVi78a}. The converse is not true.}

\begin{definition}[Zero-sum games]
	\label{def:zero-sum}
	A game $G$ is zero-sum if there are $\alpha \in\mathbb R_{>0}^n$ and $\beta\in\mathbb R$ such that $\sum_{i\in N} \alpha_i G_i + \beta \equiv 0$.
\end{definition}

We say that a player $i$ is a dummy player if $i$'s payoff is constant and all of $i$'s actions are clones.
Equivalently, $i$'s payoff is constant and no player's payoff depends on $i$'s action, irrespective of the other players' actions.
An essentially two-player game is one in which all but two players are dummies.
\begin{definition}[Essentially two-player games]
	\label{def:essentially-two-player}
	Let $G$ be a game on $A$.
	Player $i\in N$ is a dummy player if $G_i$ is constant on $A$ and all actions in $A_i$ are clones.
	$G$ is an essentially two-player game if there are distinct $j,k\in N$ such that all players in $N\setm\{j,k\}$ are dummy players.
	If $G$ is an essentially two-player game, $j,k\in N$ are distinct, and $\gamma G_j + G_k$ is constant for some $\gamma > 0$, we say that $G$ is $(j,k,\gamma)$-zero-sum.
\end{definition}

We prove that for any pair of distinct players $j,k$, there exists some $\gamma > 0$ such that $f$ returns all Nash equilibria in each $(j,k,\gamma)$-zero-sum game.

\begin{proposition}[Nash equilibria of essentially two-player zero-sum games]
	\label{prop:zero-sum-two-player}
	Let $f$ be a total, continuous, and convex-valued solution concept satisfying consistency and either
	\begin{enumerate}[label=(\roman*),ref=(\roman*)]
		\item strong consequentialism and weak rationality, or
		\item consequentialism and rationality.
	\end{enumerate}
	Then, for all distinct $j,k\in N$, there is $\gamma > 0$ such that for each $(j,k,\gamma)$-zero-sum game $G$, $\nash(G)\subseteq f(G)$.
\end{proposition}

The strategy is to construct more and more essentially two-player zero-sum games for which $f$ returns some of the Nash equilibria.
By consistency, we have succeeded if for each of the games $G$ in \Cref{prop:zero-sum-two-player} and each Nash equilibrium $p$ of $G$, $G$ is a convex combination of games for which $p$ is a Nash equilibrium and we have established that $f$ returns $p$.

\subsection{Nash equilibria of bistochastic essentially two-player zero-sum games}

The first step is to find some nontrivial game for which $f$ returns a full-support Nash equilibrium.
We use the game in which two players play a matching pennies game and all other players are dummies.

\begin{definition}[Matching pennies games]
	For distinct $j,k\in N$ and $z\in\mathbb R_{\ge 0}^{\{j,k\}}$, the matching pennies game between $j$ and $k$ with stakes $z$ is the game $G^{j,k,z}$ on $A$ with $A_j = A_k = \{1,2\}$ and $A_i=\{0\}$ for each $i\neq j,k$, where for each $a\in A$,
	\begin{align*}
		G_j(a) = \begin{cases}
			z_j \quad&\text{if } a_j+a_k \equiv 1 \pmod{2}\text{,}\\
			0 &\text{otherwise,}
		\end{cases}
		\quad\text{and}\quad
		G_k(a) = \begin{cases}
			z_k \quad&\text{if } a_j+a_k \equiv 0 \pmod{2}\text{,}\\
			0 &\text{otherwise,}
		\end{cases}
	\end{align*}
	and $G_i\equiv 0$ for each $i\neq j,k$.
	If $z_j = z_k = 1$, we write $G^{j,k}$ for short.
\end{definition}

For example, for $n=2$,
	\[
	G^{1,2}=
	\begin{pmatrix}
0,1 & 1,0\\
1,0 & 0,1\\
\end{pmatrix}
	\]

The unique coarse correlated equilibrium of any matching pennies game for $n=2$ is uniform randomization over all action profiles.
We prove that this uniform distribution is returned.
The proof uses that $f$ is a refinement of \nash when restricting to strategies with independent randomization.
\begin{lemma}[\citealp{BrBr23a}]
	\label{lem:brbr23a}
	Let $f$ be a total solution concept that satisfies consistency, consequentialism, and weak rationality. 
	Let $G$ be a game on $A$, and let $p = p_1\otimes \dots\otimes p_n$ with $p_i \in \Delta(A_i)$. 
	Then, $p\in f(G)$ implies $p \in \nash(G)$.
\end{lemma}
\begin{proof}
	Theorem 1 of \citet{BrBr23a} shows that \nash is the unique total solution concept returning strategies with independent randomization that satisfies consistency, consequentialism, and weak rationality.
	The proof of their Lemma 1 (used in the proof of their Theorem 1) remains true for solution concepts that return correlated strategies. 
	The proof of the part $f\subseteq\nash$ of their Theorem~1 assumes that there is $p = p_1\otimes\dots\otimes p_n \in f(G)\setminus \nash(G)$ and derives a contradiction. 
	Their arguments remain true for solution concepts that return correlated strategies. 
	Hence, the claim follows.
\end{proof}

Matching pennies games have a useful symmetry---a permutation of the action sets under which the game is invariant.
Indeed, switching the labels of player $1$'s actions and player $2$'s actions leaves the game unchanged.
Totality, equivariance, and convex-valuedness imply that $f$ returns a correlated strategy that is invariant under this symmetry. 
Moreover, the matching pennies game with $z_j = 0$ or $z_k= 0$ has two pure Nash equilibria. 
Together with \Cref{prop:pure-nash} and \Cref{lem:brbr23a}, this implies that $f$ returns the uniform distribution in a matching pennies game for some stakes.

\begin{lemma}[Matching pennies games]
	\label{lem:two-player-matching-pennies}
	Let $j,k\in N$ be distinct and let $f$ be a total, continuous, and convex-valued solution concept satisfying consistency and either
	\begin{enumerate}[label=(\roman*),ref=(\roman*)]
		\item strong consequentialism and weak rationality, or
		\item consequentialism and rationality.
	\end{enumerate}
	Then, there is $z\in\mathbb R_{>0}^{\{j,k\}}$ such that for each $\alpha > 0$, $\uni(A) \in f(G^{j,k,\alpha z})$, where $A_j=A_k=\{1,2\}$ and $A_i = \{0\}$ for each $i\neq j,k$.
\end{lemma}
\begin{proof}
	We prove the statement for $n = 2$ and $(j,k) = (1,2)$ for convenience.
	The proof of the general case is the same up to notational changes.
	By \Cref{lem:homogeneous-core}, we may assume that $f$ is positively homogeneous.
	
	Let $G = G^{j,k}$, and let $p\in f(G)$, which exists by totality.
	Let $\pi_*\in\Pi(\{1,2\})$ be the permutation that swaps $1$ and $2$, and let $\pi = (\pi_*,\pi_*)$.
	Note that $G = G\circ \pi$.
	Hence, by equivariance, $p\circ \pi \in f(G\circ\pi) = f(G)$.
	Thus, convex-valuedness implies that $\tilde p = \frac12 p + \frac12 (p\circ\pi) \in f(G)$.
	Observe that $\tilde p(1,1) = \tilde p(2,2)$ and $\tilde p(1,2) = \tilde p(2,1)$, and so there are $s,t\in[0,1]$ such that $\tilde p = p^{s,t}$, where
	\[
	p^{s,t} = 
	\begin{pmatrix}
s & t\\
t & s
\end{pmatrix}
	\]
	Therefore, if there exist $p',p''\in f(G)$ such that $p'(1,1) + p'(2,2) \ge \frac12$ and $p''(1,1) + p''(2,2) \le \frac12$, then $\uni(A) = p^{\frac14,\frac14} \in f(G)$ by convex-valuedness.

	If $s = t$, the statement follows since $f$ is positively homogeneous. 
	Thus, assume without loss of generality that $s > t$.
	
	For $\alpha,\beta \ge 0$, let $G^{(\alpha,\beta)} = G^{j,k,(\alpha,\beta)}$.
		\[
	G^{(\alpha,\beta)}=
	\begin{pmatrix}
0,\beta & \alpha,0\\
\alpha,0 & 0,\beta
\end{pmatrix}
	\]
	We prove that there is $\beta\in[0,1]$ such that $\uni(A) \in f(G^{(1,\beta)})$.
	For each $\beta\in[0,1]$, there are $s_\beta,t_\beta\in[0,1]$ such that $p^{s_\beta,t_\beta}\in f(G^{(1,\beta)})$.
	This follows from totality, equivariance, and convex-valuedness as above.
	Assume for contradiction that for each $\beta\in[0,1]$, $\uni(A) = p^{\frac14,\frac14}\not\in f(G^{(1,\beta)})$.
	Then, for each $\beta\in[0,1]$, one of the following holds: (i) for each $p\in f(G^{(1,\beta)})$, $p(1,1) + p(2,2) > \frac12$ or (ii) for each $p\in f(G^{(1,\beta)})$, $p(1,1) + p(2,2) < \frac12$.
	By continuity and the assumption that $s > t$, (i) holds for each $\beta\in[0,1]$.
	In particular, $s_0 > t_0$.
	
	By \Cref{prop:pure-nash}, $\chi_{(1,2)},\chi_{(2,1)} \in f(G^{(1,0)})$ since these action profiles are pure Nash equilibria.
	We claim that there is a product distribution 
	$q\in f(G^{(1,0)})\setminus \nash(G^{(1,0)})$.
	If $t_0>0$, let $\lambda=4t_0\in(0,1)$ and set $q=\lambda p^{s_0,t_0}+(1-\lambda)\chi_{(1,2)}$.
	A direct calculation shows that $q$ is a product distribution. 
	Moreover, $q_2(1)<q_2(2)$, while $q_1(2)>0$, so $q_1$ is not a best response to $q_2$ in $G^{(1,0)}$.
	If $t_0=0$, then $s_0=\frac12$. Set
	\[
	q
	=
	\frac49 p^{\frac12,0}
	+\frac19\chi_{(1,2)}
	+\frac49\chi_{(2,1)}
	=
	\begin{pmatrix}
	\frac29 & \frac19\\
	\frac49 & \frac29
	\end{pmatrix}
	=
	\left(\frac13,\frac23\right)\otimes
	\left(\frac23,\frac13\right).
	\]
	Thus $q$ is a product distribution. Against 
	$q_2=(\frac23,\frac13)$, player $1$'s unique best response in $G^{(1,0)}$ is action $2$, but $q_1$ assigns positive probability to action $1$. Hence $q\notin\nash(G^{(1,0)})$.
	In both cases, convex-valuedness and 
	$\chi_{(1,2)},\chi_{(2,1)},p^{s_0,t_0}\in f(G^{(1,0)})$
	imply that $q\in f(G^{(1,0)})$, contradicting \Cref{lem:brbr23a}.
	Hence, there is $\beta^*\in[0,1]$ such that $\uni(A) \in f(G^{(1,\beta^*)})$.
	
	If $\beta^* > 0$, the statement follows from positive homogeneity.
	So assume that $\beta^* = 0$ and $\uni(A)\not\in f(G^{(1,\beta)})$ for each $\beta\in(0,1]$.
	Let $\beta \in (0,1]$ and recall that $s_\beta > t_\beta$ by continuity.
	Since $\uni(A), \chi_{(1,2)}, \chi_{(2,1)} \in f(G^{(1,0)})$, convex-valuedness implies $p^{t_\beta,s_\beta} \in f(G^{(1,0)})$ (note the transposition of $s_\beta$ and $t_\beta$). 
	Let $\pi_1,\pi_2 \in \Pi(\{1,2\})$ such that $\pi_1$ swaps $1$ and $2$ and $\pi_2$ is the identity, and let $\pi = (\pi_1,\pi_2)$.
	Then, $p^{s_\beta,t_\beta} \in f(G^{(1,0)} \circ \pi)$.
	Observe that
	\begin{align*}
		\tilde G = \frac 12 G^{(1,\beta)} + \frac12 \left(G^{(1,0)}\circ \pi\right) = 
		\frac12 \begin{pmatrix}
			0,\beta & 1,0\\
			1,0 & 0,\beta\\ 
		\end{pmatrix}
		+ \frac12
		\begin{pmatrix}
			1,0 & 0,0\\
			0,0 & 1,0\\
		\end{pmatrix}
		=
		\frac12 
		\begin{pmatrix}
			1,\beta & 1,0 \\
			1,0 & 1,\beta\\	
		\end{pmatrix}
	\end{align*}
	Consistency implies that $p^{s_\beta,t_\beta} \in f(\tilde G)$.
	Then, it follows from \Cref{lem:adding-constants} and positive homogeneity that $p^{s_\beta,t_\beta} \in f(G^{(0,1)})$.
	Note that, if $p^{s_\beta,t_\beta}$ has a convergent subsequence that does not converge to $\uni(A)$ as $\beta$ converges to $0$, we obtain a contradiction to \Cref{lem:brbr23a} as above.
	Thus, it follows from continuity that $\uni(A) \in f(G^{(0,1)})$.
	Consistency and positive homogeneity then imply that $\uni(A) \in f(G^{(1,1)})$, which suffices by positive homogeneity.
	This finishes the proof.
\end{proof}

A matrix is deterministic if each entry is either $0$ or $1$, and it is bistochastic if all its entries are nonnegative and all row and column sums are $1$.
\begin{definition}
	[Deterministic and bistochastic matrices]
	Let $A = A_1\times A_2$, and let $T\colon A\to \mathbb R_{\ge 0}$.
	Then, $T$ is deterministic if $T(a) \in \{0,1\}$ for each $a\in A$, and $T$ is bistochastic if for each $i\in \{1,2\}$ and each $a_{-i}\in A_{-i}$, $\sum_{a_i \in A_i} T(a_i,a_{-i}) = 1$.
\end{definition}

Starting from a matching pennies game, we construct more games where the uniform distribution is a Nash equilibrium and $f$ returns this correlated strategy, by blowing up, permuting actions, and taking convex combinations.
The next two lemmas show that one can use these operations to reach any bistochastic matrix as the utility function of the first player up to a positive affine transformation.

\begin{lemma}
	[Building bistochastic matrices]
	\label{lem:building-fully-stochastic-matrices}
	Let $k,m\in\mathbb N$ such that $k \ge 2$ and $m$ is a multiple of $k$, and let $T\colon[m]^2\to\mathbb R_{\ge 0}$ be bistochastic. 
	Let $\tilde T\colon[m]^2\to\mathbb R_{\ge 0}$ be obtained from a deterministic bistochastic matrix on $[k]^2$ by replacing, for each $i\in\{1,2\}$, each $a_i\in[k]$ by $m/k$ clones.
	Then, there are $\lambda\in\Delta(\Pi([m])^2)$, $\alpha > 0$, and $\beta\in\mathbb R$ such that $\alpha T + \beta = \sum_{\pi\in \Pi([m])^2} \lambda_{\pi} \tilde T \circ \pi$.
\end{lemma}
\begin{proof}
	First, consider the case that $T$ is deterministic (i.e., a permutation matrix).
	Without loss of generality, $T$ is the identity matrix, and $\tilde T$ is obtained by replacing each $1$ in the $k\times k$ identity matrix by a block of $1$'s of size $m/k\times m/k$.
	Let $\tilde\lambda\in\mathbb R_{\ge 0}^{\Pi([m])^2}$ be a vector that assigns a weight to each pair of permutations of $[m]$ such that $\tilde\lambda_\pi = 1$ if $\supp(T) \subseteq\supp(\tilde T\circ \pi)$ and $\tilde\lambda_\pi = 0$ otherwise.
	That is, $\tilde\lambda_\pi= 1$ if and only if $\pi = (\pi_1,\pi_1\circ\pi_2)$ for some $\pi_1,\pi_2 \in \Pi([m])$ such that $\pi_2$ fixes each of the sets $\{1,\dots,m/k\},\dots,\{m-m/k+1,\dots,m\}$.
	Note that $|\tilde\lambda| = m!((m/k)!)^k$, and let $\lambda = \tilde\lambda/|\tilde\lambda|$.
	Let $\hat T = \sum_{\pi} \lambda_\pi \tilde T\circ \pi$.
	For each $a\in [m]^2$, $\hat T(a) = 1$ if $T(a) = 1$, and $\hat T(a) = \frac{m/k-1}{m-1}$ if $T(a) = 0$.
	Thus, letting $\alpha = 1 - \frac{m/k-1}{m-1}$ and $\beta = \frac{m/k-1}{m-1}$ yields the required expression.
	The case when $T$ is not necessarily deterministic follows from the fact that every bistochastic matrix is a convex combination of permutation matrices, and that convex combinations of affine transformations of a matrix are again affine transformations of that matrix.
\end{proof}

The case of \Cref{lem:building-fully-stochastic-matrices} where $k = m$ and $\beta$ is required to be $0$ (and thus $\alpha = 1$) is the Birkhoff-von Neumann theorem.
Thus, \Cref{lem:building-fully-stochastic-matrices} can be viewed as a variant of that theorem, where the basic matrices are blow-ups of smaller permutation matrices rather than honest permutation matrices.

Using \Cref{lem:building-fully-stochastic-matrices}, we show that for any bistochastic matrix $T$, there exists an essentially two-player zero-sum game in which the utility function of one non-dummy player is $T$ and the uniform distribution over all action profiles is a Nash equilibrium and is returned by $f$.

\begin{lemma}
	[Uniform equilibria of essentially two-player zero-sum games I]
	\label{lem:fully-stochastic-games}
	Let $f$ be a total, continuous, and convex-valued solution concept satisfying consistency and either
	\begin{enumerate}[label=(\roman*),ref=(\roman*)]
		\item strong consequentialism and weak rationality, or
		\item consequentialism and rationality.
	\end{enumerate}
	Let $m,n\in\mathbb N$ and let $j,k\in N$ be distinct.
	Then, there is $\gamma > 0$ such that for each bistochastic $T\colon [m]^{\{j,k\}}\to\mathbb R_{\ge 0}$, there exists a game $G$ on $A = [m]^{\{j,k\}} \times \{0\}^{N\setm \{j,k\}}$ such that for each $a_{-j,k}\in A_{-j,k}$, $G_j(\cdot,a_{-j,k}) = T$, $G_k(\cdot,a_{-j,k}) = -\gamma T$, for each $i\neq j,k$, $G_i\equiv 0$, and $\uni(A) \in f(G)\cap \nash(G)$.
\end{lemma}
\begin{proof}
	We prove the statement assuming $N = \{j,k\}$.
	The proof of the general case, i.e., with $n-2$ dummy players, is the same, albeit more notationally heavy.
	Replacing $m$ by $2m$ and using consequentialism, we may assume that $2$ divides $m$.
	Moreover, by \Cref{lem:homogeneous-core}, we may assume that $f$ is positively homogeneous.
	
	By \Cref{lem:two-player-matching-pennies}, there is $z\in\mathbb R_{>0}^{\{j,k\}}$ such that $\uni(A) \in f(G^{j,k,\alpha z})$ for each $\alpha > 0$.
	It follows that there is $\gamma > 0$ such that for $z = (1,\gamma)$, $\uni(A) \in f(G^{j,k,z})$.
	We write $G^z$ instead of $G^{j,k,z}$ for short.
	Let $\tilde G^z$ be the game on $A$ obtained from $G^z$ by replacing each action in $A_j$ and $A_k$ by $m/2$ clones.
	Consequentialism implies that $\uni(A) \in f(\tilde G^z)$.
	Since $G_j^z\colon A\to\mathbb R$ is deterministic and bistochastic, it follows from \Cref{lem:building-fully-stochastic-matrices} that there are $\alpha > 0$, $\beta\in\mathbb R$, and $\lambda\in\Delta(\Pi([m])^{\{j,k\}})$ such that
	\[
	\alpha T + \beta = \sum_{\pi\in \Pi([m])^{\{j,k\}}} \lambda_{\pi}\, (\tilde G_j^z \circ \pi).
	\]
	Let $\bar G = \sum_{\pi\in \Pi([m])^{\{j,k\}}} \lambda_{\pi}\, (\tilde G^z \circ \pi)$.
	Equivariance and consistency imply that $\uni(A) \in f(\bar G)$.
	Letting $G^\beta$ be the game with $G^\beta_j\equiv \beta$ and $G^\beta_k \equiv \gamma(1-\beta)$, and letting $G = \frac1\alpha (\bar G - G^\beta)$, we have that $G_j = T$ and $G_k = -\gamma T$.
	Moreover, it follows from \Cref{lem:adding-constants} and positive homogeneity that $\uni(A) \in f(G)$, and since $T$ is bistochastic, $\uni(A) \in \nash(G)$.
\end{proof}

We prove a weaker version of \Cref{prop:zero-sum-two-player}: any total, continuous, and convex-valued solution concept that satisfies consistency and either (i) strong consequentialism and weak rationality or (ii) consequentialism and rationality returns the uniform distribution over all action profiles in any essentially two-player zero-sum game in which this correlated strategy is a Nash equilibrium, assuming the payoffs of the two active players have a fixed ratio.

\begin{lemma}[Uniform equilibria of essentially two-player zero-sum games II]\label{lem:essentially-two-player-uniform}
	Let $f$ be a total, continuous, and convex-valued solution concept satisfying consistency and either
	\begin{enumerate}[label=(\roman*),ref=(\roman*)]
		\item strong consequentialism and weak rationality, or
		\item consequentialism and rationality.
	\end{enumerate}
	Let $j,k\in N$ be distinct, let $\gamma>0$ be as promised by \Cref{lem:fully-stochastic-games}, and let $G$ be a $(j,k,\gamma)$-zero-sum game on $A$ with $\uni(A) \in \nash(G)$. 
	Then, $\uni(A) \in f(G)$.
\end{lemma}
\begin{proof}
	To simplify notation, we assume that $N = \{j,k\}$; the proof extends straightforwardly to the general situation.
	We may assume that $f$ is positively homogeneous by \Cref{lem:homogeneous-core} and that $|A_j| = |A_k|$  by consequentialism.
	Moreover, we may assume that $G_j(a) \ge 0$ for each $a \in A$ and that $\gamma G_j + G_k\equiv 0$ by \Cref{lem:adding-constants}.

	Observe that all row sums of $G_j$ are equal and all its column sums are equal since $\uni(A) \in \nash(G)$.
	That is, $\sum_{a_k \in A_k} G_j(a_j,a_k)$ is independent of $a_j\in A_j$, and $\sum_{a_j \in A_j} G_k(a_j,a_k) = - \gamma\sum_{a_j \in A_j} G_j(a_j,a_k)$ is constant on $A_k$.
	Hence, $G_j$ is a nonnegative multiple of a bistochastic matrix.
	If $G_j\equiv 0$, then also $G_k\equiv 0$, so that all actions in $A_j$ are clones and all actions in $A_k$ are clones.
	Thus, $\uni(A)\in f(G)$ follows from consequentialism.
	If $G_j\not\equiv 0$, by positive homogeneity, we may assume that $G_j$ is bistochastic.
	It follows from \Cref{lem:fully-stochastic-games} that there is a $G'$ on $A$ with $G_j' = G_j$, $G_k' = -\gamma G_j'$, and $\uni(A) \in f(G')\cap \nash(G')$.
	But then $G = G'$, and $\uni(A) \in f(G)$ follows.
\end{proof}

\subsection{Reduction to uniform full support Nash equilibria}\label{sec:fsubsetnash}

\Cref{lem:essentially-two-player-uniform} shows that the statement of \Cref{prop:zero-sum-two-player} holds for the uniform distribution over all action profiles. 
We extend it to arbitrary Nash equilibria in two steps: first, to rational-valued Nash equilibria that do not necessarily have full support; second, to all Nash equilibria.

\begin{lemma}[Reduction to uniform Nash equilibria]
	\label{lem:full-support-to-non-full-support}
	Let $f$ be a total, continuous, and convex-valued solution concept satisfying consistency and either
	\begin{enumerate}[label=(\roman*),ref=(\roman*)]
		\item strong consequentialism and weak rationality, or
		\item consequentialism and rationality.
	\end{enumerate}
	Let $j,k\in N$ be distinct, and let $\gamma > 0$.
	Assume that for each action profile set $A\in\mathcal F(A)^n$ and each $(j,k,\gamma)$-zero-sum game $G$ on $A$ with $\uni(A)\in\nash(G)$, $\uni(A) \in f(G)$.
		Then, for each action profile set $A$ and each $(j,k,\gamma)$-zero-sum game $G$ on $A$, $\nashq(G)\subseteq f(G)$.
\end{lemma}
\begin{proof}
	Assume that $N = \{j,k\}$; the general case with $n-2$ dummy players follows similarly.
	By \Cref{lem:homogeneous-core}, we may assume that $f$ is positively homogeneous.

	First, observe that for any $(j,k,\gamma)$-zero-sum game $G$ and any $p\in \nashq( G)\cap \nashfs( G)$, it follows directly from consequentialism that $p \in f(G)$.
	Indeed, by replacing each action $a_i$ of each player $i$ by a number of clones proportional to $p_i(a_i)$ and using consequentialism, one can reduce to the case where $p$ is the uniform distribution, which is covered by assumption.
	The remainder of the proof removes the assumption that $p$ has full support.
	
	Let $A$ be an action profile set and let $G$ be a $(j,k,\gamma)$-zero-sum game on $A$ and $p\in\nashq( G)$.
	Recall that $p_j$ and $p_k$ are the players' marginals, so that $p = p_j\otimes p_k$.
	Let $\tilde N = \{i\in N\colon \supp(p_i) \subsetneq A_i\}$ be the set of players whose strategy does not have full support, and let $\tilde n = |\tilde N|$.
	We prove that $p\in f(G)$ by induction (with at most two steps).
	The base case $\tilde n = 0$ holds by assumption. 
	Now assume that $\tilde n > 0$ and the statement holds for all smaller values of $\tilde n$.
	Assume without loss of generality that $j\in \tilde N$. 
\setcounter{step}{0}
\begin{step}[$G_j(a_j,p_k) = 0$ for each $a_j\in A_j$]
	We first assume $G_j(\cdot,p_k) \equiv 0$, i.e., player $j$'s expected payoff given $p$ is $0$ for each action.
	Let $b_j\in \supp(p_j)$, and let $\tilde G$ be the game on $\tilde A$, where $\tilde A_j = A_j\cup \bigcup_{a_j\in A_j\setm \supp(p_j)}\{\tilde a_j,b_j^{a_j}\}$ for distinct $\tilde a_j,b_j^{a_j}\in U\setminus A_j$, $\tilde A_k = A_k$, and for each $i\in \{j,k\}$ and $c\in A$,
	\begin{align*}
		\tilde G_i(c) = 
		\begin{cases}
			G_i(c) \quad&\text{if } c_j \in A_j\\
			G_i(a_j,c_{-j}) &\text{if } c_j = \tilde a_j \text{ for some } a_j\in A_j\setm\supp(p_j)\\
			2G_i(b_j,c_{-j}) - G_i(a_j,c_{-j}) &\text{if } c_j = b_j^{a_j} \text{ for some } a_j\in A_j\setm\supp(p_j)
		\end{cases}
	\end{align*}
		That is, $\tilde G$ is obtained from $G$ by adding, for each action $a_j \in A_j\setm \supp(p_j)$, a clone $\tilde a_j$ of $a_j$ and a new action $b_j^{a_j}$ that corresponds to a combination of $b_j$ and $a_j$.
	Note that $\tilde G$ is a $(j,k,\gamma)$-zero-sum game since $ G$ is.
	
	Let $\tilde p_j\in\Delta(\tilde A_j)\cap\mathbb Q^{\tilde A_j}$ such that $\supp(\tilde p_j) = \tilde A_j$, for each $a_j\in \supp(p_j)\setm \{b_j\}$, $\tilde p_j(a_j) = p_j(a_j)$, and for each $a_j\in A_j\setm\supp(p_j)$, $\tilde p_j(a_j) + \tilde p_j(\tilde a_j) = \tilde p_j(b_j^{a_j})$.
	Let $\tilde p = \tilde p_j\otimes p_k$ be the product distribution of $\tilde p_j$ and $p_k$, and observe that $\tilde p \in \nashq(\tilde G)$.
	Thus, by the induction hypothesis $\tilde p \in f(\tilde G)$.
	Let $\hat p_j \in \Delta(\tilde A_j)$ such that $\hat p_j(c_j) = \tilde p_j(c_j)$ for each $c_j \in \supp(p_j) \cup\{b_j^{a_j}\colon a_j\in A_j\setminus \supp(p_j)\}$, and $\hat p_j(\tilde a_j) = \tilde p_j(a_j) + \tilde p_j(\tilde a_j)$ for each $a_j \in A_j\setminus\supp(p_j)$.
	In other words, $\hat p_j$ is obtained from $\tilde p_j$ by shifting the probability on $a_j$ to $\tilde a_j$.
	Let $\hat p = \hat p_j\otimes p_k \in\Delta(\tilde A)$.
	Consequentialism implies that $\hat p \in f(\tilde G)$.
	
	Let $\pi_j\in\Pi(\tilde A_j)$ be the permutation that swaps $\tilde a_j$ and $b_j^{a_j}$ for each $a_j\in A_j\setm \supp(p_j)$, let $\pi_k$ be the identity on $\tilde A_k$, and let $\pi = (\pi_j,\pi_k)$.
	Since  $\hat p_j(b_j^{a_j}) = \hat p_j(\tilde a_j)$ for each $a_j\in A_j\setm \supp(p_j)$, $\hat p = \hat p\circ \pi$, and so by equivariance, $\hat p \in f(\tilde G\circ \pi)$.
	Let $\hat G = \frac12 \tilde G + \frac12 \tilde G\circ\pi$.
	Consistency implies that $\hat p \in f(\hat G)$.
	Observe that, by construction of $\tilde G_j$, $\hat  G$ is obtained from $ G$ by adding two clones of $b_j$ for each $a_j\in A_j\setm \supp(p_j)$.
	Since $\hat p_j(a_j) = p_j(a_j)$ for each $a_j\in A_j\setm\{b_j\}$, it follows from consequentialism that $p\in f( G)$.
	\label{step:reduction-to-uniform1}
\end{step}
\begin{step}[$G_j(a_j,p_{-j}) \neq 0$ for some $a_j\in A_j$]\label{step:reduction-to-uniform2}
	For each $a_j\in A_j$, let $\alpha_{a_j} = G_j(a_j,p_k)$ be player $j$'s expected payoff for $a_j$ against $p_k$.
	Note that $\alpha_{a_j} \ge \alpha_{b_j}$ for each $a_j\in\supp(p_j)$ and $b_j\in A_j$.
	Let $\tilde G$ be the game on $A$ such that for each $a_j\in A_j$, $\tilde G_j(a_j,\cdot) = -\frac1\gamma\tilde G_k(a_j,\cdot) \equiv \alpha_{a_j}$. 
	Note that $\tilde G$ is a $(j,k,\gamma)$-zero-sum game.
	Moreover, all actions in $B_j = \{a_j\in A_j\colon \alpha_{a_j} \ge \alpha_{b_j}$ for each $b_j\in A_j\}$ are clones, and all actions in $A_k$ are clones. 
	Consequentialism and \Cref{prop:pure-nash} thus imply that $p\in f(\tilde G)$.	

	Let $\bar G =  G - \tilde G$.
	Then, $\bar G$ is $(j,k,\gamma)$-zero-sum, $p\in\nashq(\bar G)$, and for each $a_j\in A_j$, $\bar G_j(a_j,p_k) = \alpha_{a_j} - \alpha_{a_j} = 0$.
	Hence, by \Cref{step:reduction-to-uniform1}, $p\in f(\bar G)$.
	Consistency and positive homogeneity imply that $p\in f( G)$.
	This completes the proof.
\end{step}
\end{proof}
\begin{lemma}[Reduction to rational-valued Nash equilibria]
	\label{lem:rational-nash-to-all-nash}	
	Let $f$ be a continuous and convex-valued solution concept, let $j,k\in N$ be distinct, and let $\gamma > 0$.
	Assume that for each $(j,k,\gamma)$-zero-sum game $ G$ on $A$, $\nashq( G)\subseteq f( G)$.
		Then, for each $(j,k,\gamma)$-zero-sum game $ G$ on $A$, $\nash( G)\subseteq f( G)$.
\end{lemma}
\begin{proof}
	Let $G$ be a $(j,k,\gamma)$-zero-sum game on $A$, and let $p\in\nash(G)$.
	If $j=k$, then all payoffs in $G$ are constant. In that case every product distribution is a Nash equilibrium, and every such product distribution is the limit of rational-valued product distributions. Since $\nashq(G)\subseteq f(G)$ by assumption and $f$ is continuous, it follows that $p\in f(G)$. Hence we may assume that $j\neq k$.

	Since $G$ is $(j,k,\gamma)$-zero-sum, all players other than $j$ and $k$ are dummy players. Thus the active payoffs do not depend on the dummy players' actions, and there is a constant $c\in\mathbb R$ such that $G_k=c-\gamma G_j$. Write $H$ for player $j$'s payoff as a function on $A_j\times A_k$. Since $p\in\nash(G)$, we can write $p=p_1\otimes\cdots\otimes p_n$. Let $S=\supp(p_j)$, let $T=\supp(p_k)$, and let $v=H(p_j,p_k)$. The Nash inequalities say that $H(a_j,p_k)=v$ for each $a_j\in S$, $H(a_j,p_k)\le v$ for each $a_j\notin S$, $H(p_j,a_k)=v$ for each $a_k\in T$, and $H(p_j,a_k)\ge v$ for each $a_k\notin T$.

	Choose rational-valued mixed strategies $p_j^\ell\in\Delta(S)$ and $p_k^\ell\in\Delta(T)$ with $\supp(p_j^\ell)=S$, $\supp(p_k^\ell)=T$, $p_j^\ell\to p_j$, and $p_k^\ell\to p_k$. For every dummy player $i\neq j,k$, choose rational-valued $p_i^\ell\in\Delta(A_i)$ with $p_i^\ell\to p_i$. Put $v^\ell=H(p_j^\ell,p_k^\ell)$, and choose $\varepsilon_\ell>0$ with $\varepsilon_\ell\to0$.

	For each $a_j\in A_j$, define $r_{a_j}^\ell$ as follows: if $a_j\in S$, set $r_{a_j}^\ell=v^\ell-H(a_j,p_k^\ell)$; if $a_j\notin S$, set $r_{a_j}^\ell=-\max\{H(a_j,p_k^\ell)-v^\ell,0\}-\varepsilon_\ell$. For each $a_k\in A_k$, define $s_{a_k}^\ell$ as follows: if $a_k\in T$, set $s_{a_k}^\ell=v^\ell-H(p_j^\ell,a_k)$; if $a_k\notin T$, set $s_{a_k}^\ell=\max\{v^\ell-H(p_j^\ell,a_k),0\}+\varepsilon_\ell$. Since $p_j^\ell$ is supported on $S$ and $p_k^\ell$ is supported on $T$, we have $\sum_{a_j}p_j^\ell(a_j)r_{a_j}^\ell=0$ and $\sum_{a_k}p_k^\ell(a_k)s_{a_k}^\ell=0$. Moreover, the Nash inequalities for $p$ imply that $r_{a_j}^\ell\to0$ for each $a_j$ and $s_{a_k}^\ell\to0$ for each $a_k$.

	Define $H^\ell(a_j,a_k)=H(a_j,a_k)+r_{a_j}^\ell+s_{a_k}^\ell$. Let $G^\ell$ be the game on $A$ with $G_j^\ell(a)=H^\ell(a_j,a_k)$, $G_k^\ell(a)=c-\gamma H^\ell(a_j,a_k)$, and $G_i^\ell=G_i$ for each $i\neq j,k$. Then $G^\ell$ is $(j,k,\gamma)$-zero-sum and $G^\ell\to G$.

	Let $p^\ell=p_1^\ell\otimes\cdots\otimes p_n^\ell$. We claim that $p^\ell\in\nash(G^\ell)$. Indeed, for $a_j\in S$, the equality $\sum_{a_k}p_k^\ell(a_k)s_{a_k}^\ell=0$ gives $H^\ell(a_j,p_k^\ell)=v^\ell$, while for $a_j\notin S$ it gives $H^\ell(a_j,p_k^\ell)\le v^\ell$. Hence $p_j^\ell$ is a best response to $p_k^\ell$ for player $j$. Similarly, for $a_k\in T$, the equality $\sum_{a_j}p_j^\ell(a_j)r_{a_j}^\ell=0$ gives $H^\ell(p_j^\ell,a_k)=v^\ell$, while for $a_k\notin T$ it gives $H^\ell(p_j^\ell,a_k)\ge v^\ell$. Since player $k$'s payoff is $c-\gamma H^\ell$ with $\gamma>0$, this means that $p_k^\ell$ is a best response to $p_j^\ell$ for player $k$. All other players are dummy players, so they are also best responding.

	Thus $p^\ell\in\nash(G^\ell)$. Since all marginals of $p^\ell$ are rational-valued, $p^\ell\in\nashq(G^\ell)$. By assumption, $p^\ell\in f(G^\ell)$ for each $\ell$. Since $G^\ell\to G$, $p^\ell\to p$, and $f$ is continuous, it follows that $p\in f(G)$.
\end{proof}

\begin{proof}[Proof of \Cref{prop:zero-sum-two-player}]
	By \Cref{lem:essentially-two-player-uniform}, for any distinct $j,k\in N$, there is $\gamma > 0$ such that for each $(j,k,\gamma)$-zero-sum game $G$ on $A$ with $\uni(A) \in \nash(G)$, $\uni(A)\in f(G)$.
	Then, by \Cref{lem:full-support-to-non-full-support} and \Cref{lem:rational-nash-to-all-nash}, for any $(j,k,\gamma)$-zero-sum game $G$ on $A$, $\nash(G)\subseteq f(G)$.
	This completes the proof.
\end{proof}

\section{Reduction to zero-sum games}\label{sec:nash-to-ce}

Let $f$ be a total, continuous, and convex-valued solution concept that satisfies consistency, consequentialism, and rationality.
We prove that if $f$ returns all Nash equilibria of essentially two-player zero-sum games, then $f = \ce$.
	
Let $A = A_1\times\dots\times A_n$ and $B \subset A$.
	A path of length $k$ in $B$ is a sequence $a^0,\dots,a^k\in B$ such that for each $\ell\in[k]$, $a^{\ell-1}$ and $a^{\ell}$ differ in exactly one coordinate.
For $a,b\in B$, the distance between $a$ and $b$ in $B$ is the length of a shortest path from $a$ to $b$ in $B$, or infinity if there is no such path; $a$ and $b$ are adjacent if they are at distance $1$, i.e., if they differ in exactly one coordinate.
We say that $B$ is connected if there is a path between any two elements of $B$. 

\begin{lemma}[Basic decompositions]
\label{lem:decomposition}
	Let $T\colon A\to\mathbb R$ and $p\in\Delta(A)$ such that $\langle p,T\rangle = 0$.
	If $\supp(p)$ is connected, then there exist $T^1,\dots,T^k\colon A\to\mathbb R$ such that $\sum_{\ell\in[k]} T^\ell = T$, and for each $\ell\in[k]$, $\langle p, T^\ell\rangle = 0$ and the support of $T^\ell$ consists of one element of $A$ or two adjacent elements of $\supp(p)$.
\end{lemma}
\begin{proof}
	It suffices to prove the statement for the case $\supp(T) \subseteq\supp(p)$ since $\langle p,T'\rangle = 0$ for any $T'\colon A\to \mathbb R$ that is supported on $A\setminus\supp(p)$.
	Let $B = \supp(p)$, let $B^+ = \{a\in B\colon T(a) > 0\}$, and let $B^- = \{a\in B\colon T(a) < 0\}$.
	\begin{claim}\label{claim:decomposition1}
		Let $a\in B^+$ and $b\in B^-$ minimize the distance in $B$ between elements of $B^+$ and $B^-$, and let $m$ be that distance.
		If $m > 1$, there are $v,w\colon A\to\mathbb R$ such that $T = v + w$, $\langle p, v\rangle = 0$ and $\langle p, w\rangle = 0$, the distance between $\supp(v_+)$ and $\supp(v_-)$ is smaller than $m$, the distance between $\supp(w_+)$ and $\supp(w_-)$ is smaller than $m$, $|\supp(v)|,|\supp(w)| \le |\supp(T)|$, and $\supp(v),\supp(w)\subseteq B$.
	\end{claim}
	\begin{proof}[Proof of \Cref{claim:decomposition1}]
		Let $a^0,\dots,a^m$ be a path from $a$ to $b$ in $B$.
		Note that $T(a^\ell) = 0$ for each $\ell\in[m-1]$ since there is no path from $B^+$ to $B^-$ of length at most $m-1$ in $B$.
		Recall that $\chi_a$ is the standard unit vector at $a\in A$, and define $v = T(a^0) \chi_{a^0} - T(a^0) \frac{p(a^0)}{p(a^1)} \chi_{a^1}$ and $w = T - v$.
		Then, $\langle p, v\rangle = 0$, and so $\langle p, w\rangle = 0$.
		Moreover, the shortest path from $\supp(v_+)$ to $\supp(v_-)$ in $B$ has length $1$ since $a^0$ and $a^1$ are adjacent in $B$, and $a^1,\dots,a^m$ is a path from $\supp(w_+)$ to $\supp(w_-)$ in $B$ of length $m-1$.
		Lastly, $T = v+w$, $|\supp(v)|,|\supp(w)| \le |\supp(T)|$, and $\supp(v),\supp(w)\subseteq B$ since $a^0,\dots,a^m$ is a path in $B$ and $T$ is supported in $B$. 
	\end{proof}
	
	Now we prove the statement of the lemma by induction on $|\supp(T)|$.
	The base case $|\supp(T)| = 0$ is trivial.
	Observe that $|\supp(T)| = 1$ is not possible since $\supp(T_+),\supp(T_-)\neq\emptyset$ whenever $T\not\equiv 0$, $\langle p,T\rangle = 0$, and $\supp(T) \subseteq B$.
	Now consider the case that $|\supp(T)| \ge 2$, and assume the statement holds for all smaller support sizes.
	We run a second induction on the distance between $\supp(T_+)$ and $\supp(T_-)$.
	Denote this distance by $m$.
	
	If $m = 1$, there are $a\in \supp(T_+)$ and $b\in\supp(T_-)$ that are adjacent in $B$.
	Consider first the case that $p(a)T(a) \le -p(b)T(b)$.
	Let $v = T(a)\chi_{a} - T(a)\frac{p(a)}{p(b)}\chi_{b}$ and $w = T - v$.
	Then, $\langle p, v\rangle = 0$, $\langle p, w\rangle = \langle p, T\rangle - \langle p, v\rangle = 0$, $|\supp(v)| = 2$, and $|\supp(w)| < |\supp(T)|$.
	The statement holds trivially for $v$, and it holds for $w$ by the hypothesis of the induction on the support size.
	Hence, it also holds for $T = v + w$.
	The case $p(a)T(a) \ge -p(b)T(b)$ is analogous.
	
	If $m > 1$, let $v,w$ be as obtained from \Cref{claim:decomposition1}.
	Then, the statement holds for $v$ and $w$ by the hypothesis of the induction on $m$.
	Hence, it also holds for $T = v + w$.
	This completes both inductions and thus the proof.
\end{proof}

\begin{lemma}[From Nash equilibria to correlated equilibria]\label{lem:ce-subseteq-f}
	Let $f$ be a total, continuous, and convex-valued solution concept satisfying consistency and either
	\begin{enumerate}[label=(\roman*),ref=(\roman*)]
		\item strong consequentialism and weak rationality, or
		\item consequentialism and rationality.
	\end{enumerate}
	Assume that for all distinct $j,k\in N$, there is $\gamma_{j,k} > 0$ such that $\nash( G)\subseteq f( G)$ for each $(j,k,\gamma_{j,k})$-zero-sum game $ G$. 
	Let $ G $ be a game on $A$ and $j\in N$ such that for each $i\neq j$, $G_i \equiv 0$.
	Then, $\ce(G)\subseteq f(G)$.
\end{lemma}	
\begin{proof}
	By \Cref{lem:homogeneous-core}, we may assume that $f$ is positively homogeneous.
	We establish a special case of the lemma first.
	\setcounter{claim}{0}
	\begin{claim}
		\label{claim:connected-ce}
		Let $j \in N$, $G_j'\colon A\to\mathbb R$, $b_j\in A_j$, and $p\in\Delta(A)$ such that for each $a_j\neq b_j$, $p(a_j,\cdot)\equiv 0$, $\langle p(b_j,\cdot),G_j'(a_j,\cdot)\rangle \le 0$, and $G_j'(b_j,\cdot)\equiv 0$.
		Then, there exists a game $ G$ on $A$ such that $G_j = G_j'$, for each $i\in N$, $G_i(b_j,\cdot)\equiv 0$, and $p\in f( G)$.
	\end{claim}
	Note that $p \in \ce( G)$.
\begin{proof}[Proof of \Cref{claim:connected-ce}]
The standing zero-sum Nash-inclusion hypothesis continues to hold after the passage to the homogeneous core: if $H$ is a $(j,k,\gamma_{j,k})$-zero-sum game and $r\in\nash(H)$, then $\alpha H$ is again a $(j,k,\gamma_{j,k})$-zero-sum game and $r\in\nash(\alpha H)$ for every $\alpha>0$; hence $r$ belongs to the homogeneous core at $H$.
Let $q=p(b_j,\cdot)$, and let $G^0$ be the game on $A$ such that $G^0_j=G_j'$ and $G^0_i\equiv 0$ for each $i\neq j$.
We prove the stronger statement that $p\in f(G^0)$.
If $A_j=\{b_j\}$, then $G^0$ is the zero game, and the conclusion follows from \Cref{prop:pure-nash} and convex-valuedness.

We proceed in multiple steps.
\setcounter{step}{0}
\begin{step}
First, assume that $\supp(p)$ is connected and $\langle q,G_j'(a_j,\cdot)\rangle =0$ for each $a_j\neq b_j$.
Since $p(a_j,\cdot)\equiv 0$ for each $a_j\neq b_j$, $\supp(q)$ is connected.
For each $a_j\neq b_j$, if $G_j'(a_j,\cdot)\equiv 0$, let $L(a_j)=\emptyset$.
Otherwise, apply \Cref{lem:decomposition} to $G_j'(a_j,\cdot)$ and $q$, and write
\[
G_j'(a_j,\cdot)=\sum_{\ell\in L(a_j)} \hat T^{a_j,\ell},
\]
where $L(a_j)$ is finite and, for each $\ell\in L(a_j)$, $\langle q,\hat T^{a_j,\ell}\rangle=0$ and the support of $\hat T^{a_j,\ell}$ consists either of one element of $A_{-j}$ or of two adjacent elements of $\supp(q)$.
For each $a_j\neq b_j$ and $\ell\in L(a_j)$, let $G^{a_j,\ell}$ be the game on $A$ with $G_j^{a_j,\ell}(a_j,\cdot)=\hat T^{a_j,\ell}$, $G_j^{a_j,\ell}(c_j,\cdot)\equiv 0$ for each $c_j\neq a_j$, and $G_i^{a_j,\ell}\equiv 0$ for each $i\neq j$.
We show that $p\in f(G^{a_j,\ell})$.

If $\supp(\hat T^{a_j,\ell})=\{c_{-j}\}$, then $q(c_{-j})=0$ because $\langle q,\hat T^{a_j,\ell}\rangle=0$.
For each $r_{-j}\in\supp(q)$, $\chi_{(b_j,r_{-j})}$ is a pure Nash equilibrium of $G^{a_j,\ell}$.
Thus $p\in f(G^{a_j,\ell})$ by \Cref{prop:pure-nash} and convex-valuedness.

Now suppose that $\supp(\hat T^{a_j,\ell})=\{c_{-j},d_{-j}\}$ for adjacent $c_{-j},d_{-j}\in\supp(q)$.
Let $k\neq j$ be the unique player such that $c_k\neq d_k$.
Define a game $H^{a_j,\ell}$ on $A$ by
\[
H_j^{a_j,\ell}(e_j,r_{-j})=
\begin{cases}
\hat T^{a_j,\ell}(c_{-j}) & \text{if } e_j=a_j \text{ and } r_k=c_k,\\
\hat T^{a_j,\ell}(d_{-j}) & \text{if } e_j=a_j \text{ and } r_k=d_k,\\
0 & \text{otherwise,}
\end{cases}
\]
$H_k^{a_j,\ell}=-\gamma_{j,k}H_j^{a_j,\ell}$, and $H_i^{a_j,\ell}\equiv 0$ for each $i\neq j,k$.
Then $H^{a_j,\ell}$ is a $(j,k,\gamma_{j,k})$-zero-sum game.
Let
\[
\hat p=\frac{q(c_{-j})\chi_{(b_j,c_{-j})}+q(d_{-j})\chi_{(b_j,d_{-j})}}{q(c_{-j})+q(d_{-j})}.
\]
Since $c_{-j}$ and $d_{-j}$ differ only in coordinate $k$, $\hat p$ is a product distribution.
Moreover, $\langle q,\hat T^{a_j,\ell}\rangle=0$ implies that player $j$ is indifferent between $b_j$ and $a_j$ under $\hat p$ in $H^{a_j,\ell}$; all other deviations of player $j$ yield payoff $0$, and all other players are indifferent.
Thus $\hat p\in\nash(H^{a_j,\ell})$, and hence $\hat p\in f(H^{a_j,\ell})$ by the standing zero-sum Nash-inclusion hypothesis.

Let $D^{a_j,\ell}=G^{a_j,\ell}-H^{a_j,\ell}$.
For each $r_{-j}\in\{c_{-j},d_{-j}\}$, $\chi_{(b_j,r_{-j})}$ is a pure Nash equilibrium of $D^{a_j,\ell}$: player $j$ obtains $0$ from every action, and for each $i\neq j$, every unilateral deviation from $(b_j,r_{-j})$ leaves player $i$'s payoff equal to $0$.
Therefore $\hat p\in f(D^{a_j,\ell})$ by \Cref{prop:pure-nash} and convex-valuedness.
Consistency and positive homogeneity imply $\hat p\in f(G^{a_j,\ell})$.
For each $r_{-j}\in\supp(q)\setminus\{c_{-j},d_{-j}\}$, $\chi_{(b_j,r_{-j})}$ is a pure Nash equilibrium of $G^{a_j,\ell}$.
Since
\[
p=(q(c_{-j})+q(d_{-j}))\hat p+\sum_{r_{-j}\in\supp(q)\setminus\{c_{-j},d_{-j}\}}q(r_{-j})\chi_{(b_j,r_{-j})},
\]
convex-valuedness gives $p\in f(G^{a_j,\ell})$.

Let
\[
\bar G=\sum_{a_j\neq b_j}\sum_{\ell\in L(a_j)}G^{a_j,\ell},
\]
with the convention that an empty sum is the zero game.
If the family of summands is empty, then $\bar G=G^0$ is the zero game, and $p\in f(\bar G)$ by \Cref{prop:pure-nash} and convex-valuedness.
Otherwise, by consistency and positive homogeneity, $p\in f(\bar G)$.
By construction, $\bar G=G^0$.
Thus $p\in f(G^0)$.
\label{step:connected-ce2}
\end{step}

\begin{step}
Second, still assuming that $\supp(p)$ is connected, we allow $\langle q,G_j'(a_j,\cdot)\rangle\leq 0$ for each $a_j\neq b_j$.
Let $\alpha\in\mathbb R^{A_j}$ be given by $\alpha_{a_j}=\langle q,G_j'(a_j,\cdot)\rangle$ for each $a_j\in A_j$.
Then $\alpha_{b_j}=0$ and $\alpha_{a_j}\leq 0$ for each $a_j\neq b_j$.
Let $\hat G$ be the game on $A$ such that $\hat G_j(a_j,\cdot)\equiv \alpha_{a_j}$ for each $a_j\in A_j$, and $\hat G_i\equiv 0$ for each $i\neq j$.
Let $\tilde G$ be the game on $A$ such that $\tilde G_j=G_j'-\hat G_j$ and $\tilde G_i\equiv 0$ for each $i\neq j$.
By \Cref{step:connected-ce2} applied to $G_j'-\hat G_j$, we have $p\in f(\tilde G)$.
For each $r_{-j}\in A_{-j}$, $(b_j,r_{-j})$ is a pure Nash equilibrium of $\hat G$.
Thus $p\in f(\hat G)$ by \Cref{prop:pure-nash} and convex-valuedness.
Consistency and positive homogeneity imply $p\in f(\tilde G+\hat G)=f(G^0)$.
\label{step:connected-ce3}
\end{step}

\begin{step}
Third, we remove the assumption that $\supp(p)$ is connected.
Let $u$ be the uniform distribution on $A_{-j}$, and let
\[
M=\max_{a_j\neq b_j}\left|\left\langle u,G_j'(a_j,\cdot)\right\rangle\right|.
\]
For each $m\geq 1$, let $\eta_m=1/m$ and $\rho_m=\eta_m/(2(M+1))$.
Define $q^m=(1-\rho_m)q+\rho_m u$.
Let $p^m\in\Delta(A)$ be given by $p^m(b_j,\cdot)=q^m$ and $p^m(a_j,\cdot)\equiv 0$ for each $a_j\neq b_j$.
Finally, let $G^m$ be the game on $A$ such that $G^m_i\equiv 0$ for each $i\neq j$, $G^m_j(b_j,\cdot)\equiv 0$, and
\[
G^m_j(a_j,\cdot)=G_j'(a_j,\cdot)-\eta_m
\]
for each $a_j\neq b_j$.
The support of $q^m$ is all of $A_{-j}$, which is connected; hence $p^m$ is connected.
Moreover, for each $a_j\neq b_j$,
\[
\left\langle q^m,G^m_j(a_j,\cdot)\right\rangle
=(1-\rho_m)\left\langle q,G_j'(a_j,\cdot)\right\rangle
+\rho_m\left\langle u,G_j'(a_j,\cdot)\right\rangle-\eta_m
\leq \rho_m M-\eta_m
=-\eta_m\frac{M+2}{2(M+1)}<0.
\]
Thus, by \Cref{step:connected-ce3}, $p^m\in f(G^m)$ for each $m$.
Since $p^m\to p$ and $G^m\to G^0$, continuity implies $p\in f(G^0)$.
This proves the stronger statement, and $G^0$ has all properties required in \Cref{claim:connected-ce}.
\end{step}
\end{proof}	
	
	We use \Cref{claim:connected-ce} to prove the lemma.
	Let $G$ be a game on $A$ and $j\in N$ such that for each $i\neq j$, $G_i\equiv 0$, and let $p \in \ce(G)$.
		For each $b_j\in A_j$ with $p(b_j,\cdot) \not\equiv 0$, let $p^{b_j}\in \Delta(A)$ such that $p^{b_j}(b_j,\cdot) = \frac{p(b_j,\cdot)}{\sum_{a_{-j} \in A_{-j}} |p(b_j,a_{-j})|}$.
		Thus, $p^{b_j}(b_j,\cdot)$ is the correlated strategy of players other than $j$ if $j$ receives the signal $b_j$.
		
		Fix $b_j\in A_j$ with $p(b_j,\cdot) \not\equiv 0$.
		We decompose $G$ into three types of games.
		\begin{enumerate}[label=(\roman*)]
			\item First, let $G_j'\colon A\to\mathbb R$ such that for each $a_j\in A_j$, $G_j'(a_j,\cdot) = G_j(a_j,\cdot) - G_j(b_j,\cdot)$.
		Let $G^{b_j}$ be the game on $A$ obtained by applying \Cref{claim:connected-ce} to $j$, $G_j'$, $b_j$, and $p^{b_j}$.
		The hypotheses hold since $p$ is a correlated equilibrium of $G$.
		By the claim, $G_j^{b_j} = G_j'$ and $p^{b_j} \in f(G^{b_j})$.
			
			\item Second, for each $k\neq j$, let $\bar  G^{b_j,k}$ be the game on $A$ such that $\bar G^{b_j,k}_k = - G^{b_j}_k$, and for each $i\neq k$, $\bar G^{b_j,k}_i \equiv 0$.
		Then, for each $a_{-j}\in A_{-j}$, $\chi_{(b_j,a_{-j})}$ is a pure Nash equilibrium of $\bar G^{b_j,k}$ since $\bar G^{b_j,k}_k(b_j,\cdot) \equiv 0$.
		Hence, by \Cref{prop:pure-nash} and convex-valuedness, $p^{b_j} \in f(\bar G^{b_j,k})$.
		
			\item Third, let $\tilde G^{b_j}$ be the game on $A$ such that for each $a_j \in A_j$, $\tilde G_j^{b_j}(a_j,\cdot) \equiv G_j(b_j,\cdot)$, and for each $i\neq j$, $\tilde G_i^{b_j}\equiv 0$.
		Note that each action profile is a pure Nash equilibrium of $\tilde G^{b_j}$.
		Hence, by \Cref{prop:pure-nash} and convex-valuedness, $f(\tilde G^{b_j}) = \Delta(A)$.
		In particular, $p^{b_j}\in f(\tilde G^{b_j})$.
		\end{enumerate}
		Observe that $ G =  G^{b_j} + \sum_{k\neq j}\bar G^{b_j,k} + \tilde G^{b_j}$.
		Consistency and positive homogeneity imply that $p^{b_j}\in f(G)$.
		
		Since $p^{b_j}\in f( G)$ for each $b_j\in A_j$ with $p(b_j,\cdot)\not\equiv 0$, convex-valuedness implies that $p\in f( G)$ as required.
\end{proof}

It is straightforward to remove the assumption that $G_i\equiv 0$ for all $i\neq j$ from \Cref{lem:ce-subseteq-f} using consistency.

\section{$f = \ce$}\label{sec:f-is-ce}

We prove \Cref{thm:ce}.
The first lemma shows that \ce satisfies all of the axioms.

\begin{lemma}[\ce satisfies the axioms]\label{lem:ce-is-nice}
	$\ce$ satisfies totality, continuity, convex-valuedness, consistency, consequentialism, and rationality.
\end{lemma}

\begin{proof}
	$\ce$ is total since it is a coarsening of $\nash$.
	For every game $G$ on $A$, $\ce(G)$ is defined by linear constraints, and the constraints depend linearly on $G$.
	Thus, \ce is continuous and convex-valued, and it satisfies consistency.
	$\ce$ satisfies consequentialism since adding clones only introduces redundant constraints.
	Lastly, $\ce$ satisfies rationality since if $a_i$ is dominated by $b_i$, player $i$ prefers $b_i$ to $a_i$ when recommended to play $a_i$.
\end{proof}

\begin{lemma}[Containment in \ce]
	\label{lem:f-subset-ce}
	Let $f$ be a total, continuous, and convex-valued solution concept that satisfies consistency, consequentialism, and rationality. 
	Assume that for all distinct $j,k\in N$, there is $\gamma_{j,k} > 0$ such that $\nash( G)\subseteq f( G)$ for each $(j,k,\gamma_{j,k})$-zero-sum game $ G$. 
	Then, $f\subseteq \ce$.
\end{lemma}
\begin{proof}
	Assume for contradiction that there is a game $G$ on $A$ such that $f( G)\setminus \ce( G)\neq \emptyset$, and let $p \in f( G)\setminus \ce( G)$.
	Then, there are $j\in N$, $b_j,c_j\in A_j$, and $\eps > 0$ such that
	\begin{align}
	\label{eq:f-subseteq-ce1}
		\langle p(c_j,\cdot),G_j(c_j,\cdot) - G_j(b_j,\cdot) + \eps\mathbf 1\rangle < 0
	\end{align}
	That is, $j$ can increase its payoff by more than $\eps$ by playing $b_j$ instead of $c_j$ when receiving the signal $c_j$.
	By introducing a clone of $b_j$ and using consequentialism, we may assume that $p(b_j,\cdot) \equiv 0$.
	Let 
	\begin{align*}
		A_j^- &= \{a_j\in A_j\colon \langle p(a_j,\cdot), G_j(c_j,\cdot) - G_j(b_j,\cdot) + \eps\mathbf 1 \rangle < 0\} \text{, and}\\
		A_j^+ &= \{a_j\in A_j\colon \langle p(a_j,\cdot), G_j(c_j,\cdot) - G_j(b_j,\cdot) + \eps\mathbf 1\rangle \ge 0\}
	\end{align*}
	That is, $A_j^-$ is the set of actions $a_j$ of $j$ so that if the signal is $a_j$, then the payoff for $b_j$ is higher than that of $c_j$ by more than $\eps$.
	Note that $b_j\in A_j^+$ (since $p(b_j,\cdot)\equiv 0$), and that $c_j\in A_j^-$ by \eqref{eq:f-subseteq-ce1}.
	Let $\hat G$ be the game on $A$ such that for each $a_j\in A_j^-$, $\hat G_j(a_j,\cdot)\equiv 0$, for each $a_j\in A_j^+$, $\hat G_j(a_j,\cdot) = G_j(c_j,\cdot) - G_j(b_j,\cdot) + \eps\mathbf 1$, and for each $i\neq j$, $\hat G_i\equiv 0$.
	Observe that $A_j^-$ and $A_j^+$ are sets of clones in $\hat G$, respectively.
	Then, $p\in \ce(\hat G)$ by definition of $A_j^-$ and $A_j^+$, and $p \in f(\hat G)$ by \Cref{lem:ce-subseteq-f}.
	
	To conclude, let $\bar G = \frac12 G + \frac12 \hat G$.
	Consistency implies that $p\in f(\bar G)$. 
	Note that $b_j$ dominates $c_j$ in $\bar G$.
	Since $p(c_j,\cdot)\not\equiv 0$ by \eqref{eq:f-subseteq-ce1}, this contradicts rationality.
\end{proof}

\begin{proof}[Proof of \Cref{thm:ce}]
	By \Cref{lem:ce-is-nice}, \ce satisfies the axioms.
	It remains to show that \ce is the only such solution concept.
	
	By \Cref{prop:zero-sum-two-player}, for all distinct $j,k\in N$, there is $\gamma_{j,k} > 0$ such that $\nash( G)\subseteq f( G)$ for each $(j,k,\gamma_{j,k})$-zero-sum game $ G$. 
	Then \Cref{lem:f-subset-ce} implies that $f\subseteq \ce$.
	To show that $\ce\subseteq f$, let $G$ be a game on $A$, and let $p\in \ce(G)$.
	For each $j\in N$, let $G^j$ be the game on $A$ with $G_j^j = n G_j$, and for each $i\neq j$, $G^j_i \equiv 0$.
	Note that $G = \frac1n \sum_{j\in N} G^j$, and for each $j\in N$, $p\in \ce(G^j)$.
	\Cref{lem:ce-subseteq-f} implies that $p\in f(G^j)$ for each $j\in N$.
	Thus, $p \in f(G)$ by consistency, concluding the proof.
\end{proof}

\section{Characterization of \cce}\label{sec:appendix-cce}

\begin{lemma}[\cce satisfies the axioms]\label{lem:cce-is-nice}
	\cce satisfies totality, continuity, convex-valuedness, consistency, strong consequentialism, and weak rationality.
\end{lemma}
\begin{proof}
	Totality, continuity, convex-valuedness, and consistency follow from arguments similar to those used in \Cref{lem:ce-is-nice} to show that \ce has these properties.
	\cce satisfies strong consequentialism since for any game $G$, whether a correlated strategy satisfies the constraints defining $\cce(G)$ does not depend on how probability is distributed over clones.
	Lastly, \cce satisfies weak rationality since if $b_i$ is a dominant action, player $i$'s expected utility of committing to $b_i$ is weakly higher than that of always following their recommendation for any correlated strategy.
\end{proof}

\begin{lemma}[Containment of \cce]
	\label{lem:cce-subseteq-f}
	Let $f$ be a total, continuous, and convex-valued solution concept that satisfies consistency, strong consequentialism, and weak rationality.
	Then, $\cce\subseteq f$.
\end{lemma}
\begin{proof}
	By \Cref{lem:homogeneous-core}, we may assume that $f$ is positively homogeneous.	
	Several times during the proof, we use that $\ce(G)\subseteq f(G)$ for any game $G$ by \Cref{lem:ce-subseteq-f} and \Cref{prop:zero-sum-two-player}.
	Let $G$ be a game on $A$, and let $p \in \cce(G)$.
	We prove that $p\in f(G)$.
	By consistency and positive homogeneity, we may assume that there is $i\in N$ such that $G_j\equiv 0$ for all $j\neq i$.
	
	First, consider the case that $G_i \in\mathbb Q^A$.
	Let $p$ be a vertex of the convex polytope $\cce(G)$, and note that $p\in \Delta(A)\cap \mathbb Q^A$.
	Strong consequentialism allows reducing to a simpler situation.
	\setcounter{claim}{0}
	\begin{claim}\label{claim:cce-subseteq-f1}
		We may assume that there is $B\in\mathcal F(U)$ such that for each $i\in N$, $B\subseteq A_i$, and $p = \uni(D)$, where $D = \{(b,\dots,b)\colon b \in B\}$ is the diagonal in $B^n$.
	\end{claim}
	\begin{proof}[Proof of \Cref{claim:cce-subseteq-f1}]
	Choose $k\in\mathbb N$ such that $kp(a)\in\mathbb N \cup\{0\}$ for each $a\in A$.
	Let $B\in\mathcal F(U)$ be disjoint from $\bigcup_{j\in N}A_j$ with $|B|=k$.
	Since $\sum_{a\in A} k p(a)=k$, we can fix a map $\sigma\colon B\to A$ such that for each $a \in A$,
	\[
		|\sigma^{-1}(a)| = kp(a)
	\]
	For each $j\in N$, let $\tilde A_j=A_j\cup B$ and $\tilde A = \tilde A_1\times\cdots\times \tilde A_n$.
	Let $\phi\colon\tilde A\to A$ be the surjection such that for each $j \in N$, $\phi_j(a_j)=a_j$ for each $a_j \in A_j$ and $\phi_j(b)=\sigma(b)_j$ for each $b\in B$.
	Let $\tilde G$ be the game on $\tilde A$ given by $\tilde G=G\circ\phi$, so that $\tilde G$ is a blow-up of $G$ with surjection $\phi$.
	Intuitively, for each action $a\in A$, we add $kp(a)$ clones of $a_j$ for each $j$.

	Let $\tilde p=\uni(D)$.
	For any $a\in A$,
	\[
		\phi_*(\tilde p)(a)
		=\sum_{\tilde a\in \phi^{-1}(a)} \tilde p(\tilde a)
		=\frac{1}{|B|}\,|\{b\in B:\sigma(b)=a\}|
		=\frac{k p(a)}{k}
		= p(a)
	\]
	and so $\phi_*(\tilde p)=p$.
	Since $\cce$ and $f$ both satisfy strong consequentialism, $p\in\cce(G)$ if and only if $\tilde p\in\cce(\tilde G)$, and	$p\in f(G)$ if and only if $\tilde p\in f(\tilde G)$.

	The triple $(\tilde G,\tilde A,\tilde p)$ has the properties claimed in \Cref{claim:cce-subseteq-f1}.
	Relabeling this triple to $(G,A,p)$ finishes the proof.
\end{proof}		

	If $|B| = 1$, then $p \in \ce(G)\subseteq f(G)$ and we are done.
	For the rest of the proof, assume that $|B| \ge 2$.
	We show that $G$ can be written as a sum of three types of games. 
	For each such game, $f$ returns $p$ as a consequence of $\ce\subseteq f$ and strong consequentialism.
	Consistency and homogeneity then imply that $p \in f(G)$.

	For $b^* \in B$, let $\tilde A = (\{b^*\}\cup (A_i\setminus B))\times A_{-i}$, and let $\tilde G^*$ be the game on $\tilde A$ such that 
	\begin{enumerate}[label=(\roman*)]
		\item $\tilde G^*_i(b^*,b,\dots,b) = G_i(b,\dots,b)$ for each $b \in B$,
		\item $\tilde G^*_i(b^*,a_{-i}) = 0$ for each $a_{-i}\in A_{-i}\setminus D_{-i}$,
		\item $\tilde G^*_i(a) = G_i(a)$ for each $a\in A$ with $a_i\in A_i\setminus B$, and
		\item $\tilde G^*_j\equiv 0$ for each $j\neq i$.
	\end{enumerate}
	Let $\tilde p = \uni(\tilde D)$, where $\tilde D = \{b^*\}\times D_{-i}$.
	Observe that  $i$'s payoff for $\tilde p$ in $\tilde G^*$ equals $i$'s payoff for $p$ in $G$.
	Since $p\in \cce(G)$, it follows that $\tilde p\in \ce(\tilde G^*)\subseteq f(\tilde G^*)$.
	Let $\tilde G = \tilde G^*\circ \phi$  be the blow-up of $\tilde G^*$ with surjection $\phi\colon A\to \tilde A$, where $\phi_i^{-1}(b^*) = B$, $\phi_i$ is the identity on $A_i\setminus B$, and for each $j\neq i$, $\phi_j$ is the identity on $A_j$.
	That is, $\tilde G$ is obtained from $\tilde G^*$ by replacing $b^*$ by $|B|$ clones.
	Strong consequentialism and $\tilde p \in f(\tilde G^*)$ imply that $p\in f(\tilde G)$.
	Note that $\tilde G_i$ agrees with $G_i$ on $D$ and on $(A_i\setminus B)\times A_{-i}$.
	
	Let $c,c'\in B$ be distinct, and let $\tilde G^{c}$ be the game on $\tilde A^c = (\{c,c'\}\cup (A_i\setminus B))\times A_{-i}$ such that
	\begin{enumerate}[label=(\roman*)]
		\item $\tilde G^{c}_i(c,b,\dots,b) = G_i(c,b,\dots,b) - G_i(b,\dots,b)$ for each $b \in B$,
		\item $\tilde G^{c}_i(c,a_{-i}) = 0$ for each $a_{-i}\in A_{-i}\setminus D_{-i}$,
		\item $\tilde G^{c}_i(a) = 0$ for each $a\in \tilde A^c$ with $a_i\in \{c'\}\cup (A_i\setminus B)$, and
		\item $\tilde G^{c}_j \equiv 0$ for each $j\neq i$.
	\end{enumerate}
	Let $\tilde p^c = \uni(\{(c,\dots,c)\} \cup \{(c',b,\dots,b)\colon b\in B\setminus\{c\}\})\in\Delta(\tilde A^c)$.
	Since $p\in \cce(G)$, we have $\sum_{b\in B} p(b,\dots,b) \left(G_i(b,\dots,b) - G_i(c,b,\dots,b)\right) \ge 0$.
	That is, player $i$ cannot profitably deviate to $c$ before observing their recommendation in $G$.
	Hence, player $i$ cannot profitably deviate to $c$ when their recommendation is $c'$ in $\tilde G^c$.
	Moreover, $\tilde G^c_i(a_i,c,\dots,c)= 0$ for each $a_i \in \tilde A_i^c$, and so player $i$ cannot profitably deviate to any other action when their recommendation is $c$.
	Hence, $\tilde p^c\in \ce(\tilde G^{c})\subseteq f(\tilde G^{c})$.
	Let $G^c = \tilde G^c \circ \phi$ be the blow-up of $\tilde G^c$ with surjection $\phi\colon A\to \tilde A^c$, where $\phi_i^{-1}(c') = B\setminus\{c\}$, $\phi_i$ is the identity on $(A_i\setminus B)\cup\{c\}$, and for each $j\neq i$, $\phi_j$ is the identity on $A_j$.
	Strong consequentialism implies that $p\in f(G^c)$.
	
	Lastly, let $\hat G$ be the game on $A$ with 
	\begin{enumerate}[label=(\roman*)]
		\item $\hat G_i(a) = 0$ for each $a \in A$ with $a_i\in B$ and $a_{-i}\in D_{-i}$,
		\item $\hat G_i(a) = G_i(a)$ for each $a\in A$ with $a_i\in B$ and $a_{-i}\not\in D_{-i}$,
		\item $\hat G_i(a) = 0$ for each $a\in A$ with $a_i\in A_i\setminus B$, and
		\item $\hat G_j \equiv 0$ for each $j\neq i$.
	\end{enumerate}
	Since $\hat G_i(a) = 0$ for each $a\in A$ with $a_{-i}\in D_{-i}$, it follows that $p\in \ce(\hat G)\subseteq f(\hat G)$.
	
	Observe that $G = \tilde G + \hat G + \sum_{c\in B} G^c$.
	Thus, consistency and positive homogeneity imply that $p\in f(G)$.
	This shows that $f(G)$ contains all vertices of $\cce(G)$, and thus that $\cce(G)\subseteq f(G)$ by convex-valuedness.
	
	Second, consider the general case $G_i\in\mathbb R^A$, and let $p$ be a vertex of $\cce(G)$.
	By convex-valuedness, it suffices to prove that $p \in f(G)$.
	There exist sequences of games $(G^t)_{t\in\mathbb N}$ on $A$ and $(p^t)_{t\in\mathbb N}\subseteq\Delta(A)\cap \mathbb Q^A$ such that, for each $t\in\mathbb N$, $G_i^t\in \mathbb Q^A$, $G_j^t\equiv 0$ for all $j\neq i$, and $p^t \in \cce(G^t)$ and $G_i^t\to G_i$ and $p^t\to p$.\footnote{This follows from a standard perturbation argument and the fact that $\cce(G)$ is the set of solutions to a linear system.
	Consider a polytope
	\[
		P=\{x\in\mathbb R^d:Cx=e,\ Ax\le b\},
	\]
	where, in the present application, \(Cx=e\) consists of the fixed simplex equality together with the zero-probability constraints defining the face of the simplex that contains \(p\). Let
	\[
		I=\{r:a_r^\top p=b_r\}
	\]
	be the full set of active inequalities at \(p\). Since \(p\) is a vertex, there is \(J\subseteq I\) such that the rows of \(\binom{C}{A_J}\) have rank \(d\). Choose \(p^t\in\Delta(A)\cap \mathbb Q^A\) with \(p^t\to p\), \(p^t(a)=0\) whenever \(p(a)=0\), and \(p^t(a)>0\) whenever \(p(a)>0\). For each active constraint \(r\in I\), choose rational perturbations \(a_r^t,b_r^t\) with \(a_r^t\to a_r\), \(b_r^t\to b_r\), and \(a_r^{t\top}p^t\le b_r^t\), making the constraints in \(J\) tight. For \(r\notin I\), the slack \(b_r-a_r^\top p\) is positive, so sufficiently close rational perturbations still satisfy \(a_r^{t\top}p^t<b_r^t\). The rows corresponding to \(J\) remain independent for large \(t\), so \(p^t\) is a feasible basic solution of the perturbed rational system. Applying this to the linear system defining \(\cce(G)\), with the fixed simplex constraints left unchanged and the payoff-dependent rows induced by rational payoff vectors \(G_i^t\to G_i\), gives \(p^t\in\cce(G^t)\), \(p^t\in\Delta(A)\cap \mathbb Q^A\), and \(p^t\to p\).}

	Then it follows from the first part that $p^t \in f(G^t)$ for each $t\in\mathbb N$, and thus $p\in f(G)$ by continuity.
	This finishes the proof.
\end{proof}

\begin{lemma}[Containment in \cce]
	\label{lem:f-subset-cce}
	Let $f$ be a total, continuous, and convex-valued solution concept that satisfies consistency, strong consequentialism, and weak rationality.
	Then, $f\subseteq \cce$.
\end{lemma}
\begin{proof}
	Assume for contradiction that $f\not\subseteq \cce$.
	Then, there exist a game $G$ on $A$, $p\in f(G)$, $i\in N$, $a_i^*\in A_i$, and $\eps > 0$ such that
	\begin{align}
		\sum_{a\in A} p(a)\left(G_i(a) - G_i(a_i^*,a_{-i})\right) < -\eps
		\label{eq:f-subseteq-cce1}
	\end{align}
	That is, $i$ can increase its payoff by more than $\eps$ by deviating to $a_i^*$ before observing its signal.
	We use strong consequentialism to construct a game $\tilde G$ on $\tilde A$ and a correlated strategy $\tilde p \in f(\tilde G)\setminus \cce(\tilde G)$ such that for each $a_{-i} \in \tilde A_{-i}$, there is at most one action $a_i\in \tilde A_i$ with $\tilde p(a_i,a_{-i}) > 0$.
	This step is similar to the proof of \Cref{claim:cce-subseteq-f1} in \Cref{lem:cce-subseteq-f}.
	
	Let $B \in\mathcal F(U)$ be disjoint from $\bigcup_{j\in N} A_j$ and $|B| = |A|$.
	Fix a bijection $\sigma\colon B \to A$.
	Let $\tilde A = (A_1\cup B)\times \dots \times (A_n\cup B)$, and let $\tilde G = G\circ \phi$ be the blow-up of $G$ with surjection $\phi\colon \tilde A\to A$ such that for each $j \in N$ and $a_j \in A_j$, $\phi_j^{-1}(a_j) = \{a_j\} \cup \{b \in B\colon a_j = \sigma(b)_j\}$.
	That is, for each action $a_j$ of each player $j$, we add $|A|/|A_j| = |A_{-j}|$ clones of $a_j$ labeled $\sigma^{-1}(a_j,a_{-j})$, where $a_{-j}$ ranges over $A_{-j}$.
	Let $\tilde p\in\Delta(\tilde A)$ such that for each $b\in B$, $\tilde p(b,\dots,b) = p(\sigma(b))$.
	Hence, $\tilde p$ is supported on the diagonal $D = \{(b,\dots,b)\colon b \in B\}$ of $B^n$.
	By construction, $\phi_*(\tilde p) = p$, and so $\tilde p \in f(\tilde G)$ by strong consequentialism.
	By \eqref{eq:f-subseteq-cce1},
	\begin{align}
		\sum_{b\in B} \tilde p(b,\dots,b)\left(\tilde G_i(b,\dots,b) - \tilde G_i(a_i^*,b,\dots,b)\right) < -\eps
		\label{eq:f-subseteq-cce2}
	\end{align}
	
	Let $C > \max_{a,a'\in A} \bigl(G_i(a) - G_i(a')\bigr) + \eps$ be a number larger than the maximal payoff difference for $i$ between any two action profiles plus $\eps$.	
	Let $\hat G$ be the game on $\tilde A$ such that
	\begin{enumerate}[label=(\roman*)]
		\item $\hat G_i(b,\dots,b) = \tilde G_i(a_i^*,b,\dots,b) - \tilde G_i(b,\dots,b) - \eps$ for each $b\in B$,
		\item $\hat G_i(a_i^*,\cdot) \equiv 0$,
		\item $\hat G_i(a) = -C$ for each $a\in \tilde A\setminus D$ with $a_i\neq a_i^*$, and
		\item $\hat G_j\equiv 0$ for each $j\neq i$.
	\end{enumerate}
	Then, $\tilde p\in \cce(\hat G) \subseteq f(\hat G)$ by \eqref{eq:f-subseteq-cce2} and \Cref{lem:cce-subseteq-f}.
	
	Consistency implies that $\tilde p \in f(\frac12 \tilde G + \frac12 \hat G)$.
	Note that $a_i^*$ is a dominant action for $i$ in $\frac12 \tilde G + \frac12 \hat G$, and $\tilde p(a_i^*,\cdot) \equiv 0$.
	This contradicts weak rationality.
\end{proof}

\Cref{thm:cce} is immediate from \Cref{lem:cce-subseteq-f} and \Cref{lem:f-subset-cce}.

\end{document}